\documentclass[11p,reqno]{amsart}

\usepackage[T1]{fontenc}
\usepackage{graphicx}
\usepackage{amssymb,amsthm,amsmath,mathrsfs,bm,braket,marginnote}
\usepackage{enumerate}
\usepackage{appendix}
\usepackage[colorlinks=true, pdfstartview=FitV, linkcolor=blue, citecolor=blue, urlcolor=blue]{hyperref}

\usepackage{pgf}
\usepackage{pgfplots}
\usepackage{tikz}
\usetikzlibrary{arrows,calc}
\usepackage{verbatim}
\usetikzlibrary{decorations.pathreplacing,decorations.pathmorphing}
\usepackage[numbers,sort&compress]{natbib}

\numberwithin{equation}{section}
\newtheorem{theorem}{Theorem}[section]
\newtheorem{lemma}[theorem]{Lemma}
\newtheorem{define}[theorem]{Definition}
\newtheorem{remark}[theorem]{Remark}
\newtheorem{coro}[theorem]{Corollary}

\newtheorem{prop}[theorem]{Property}

\newcommand{\rank}[1]{\text{rank}{#1}}
 \reversemarginpar
\begin{document}

\title[Partial-skew-orthogonal polynomials and integrable lattices]{Partial-skew-orthogonal  polynomials and related integrable lattices with Pfaffian tau-functions}
%Moment reduction for multipeakons from the CH equation to a two-component mCH equation and their nonisospectral generalizations
%    Information for first author
%\affiliation{LSEC, Institute of Computational Mathematics and Scientific Engineering Computing, AMSS, Chinese Academy of Sciences, P.O.Box 2719, Beijing 100190, PR China.}{lsec}

%    Address of record for the research reported here
\author[Xiang-Ke Chang]{Xiang-Ke Chang $^{\dag}$}

\thanks{$^\ast$Corresponding author: Shi-Hao Li (lishihao@lsec.cc.ac.cn)}

\thanks{$^\dag$LSEC, ICMSEC, Academy of Mathematics and Systems Science, Chinese Academy of Sciences, P.O.Box 2719, Beijing 100190, PR China}
% lishihao@lsec.cc.ac.cn

%    Information for second author
\author[Yi He]{Yi He $^{\ddag}$}
%    Address of record for the research reported here
\thanks{$^\S$School of Mathematical Sciences, University of Chinese Academy of Sciences, Beijing 100049, PR China}
\thanks{$^\ddag$Wuhan Institute of Physics and Mathematics, Chinese Academy of
  Sciences, Wuhan 430071, PR China}
%heyi@wipm.ac.cn
\author[Xing-Biao Hu]{Xing-Biao Hu $^{\dag,\S}$}
\author[Shi-Hao Li]{Shi-Hao Li $^{\dag,\S,\ast}$}

%    Address of record for the research reported here
%hxb@lsec.cc.ac.cn

\thanks{Email addresses: changxk@lsec.cc.ac.cn; heyi@wipm.ac.cn; hxb@lsec.cc.ac.cn; lishihao@lsec.cc.ac.cn}

%    \thanks will become a 1st page footnote.

%    Information for third author

%    \thanks will become a 1st page footnote.

%    General info
\date{}

\dedicatory{}


\begin{abstract}
Skew-orthogonal polynomials (SOPs) arise in the study of the $n$-point distribution function for orthogonal and symplectic random matrix ensembles.
Motivated by the average of characteristic polynomials of the Bures random matrix ensemble studied in \cite{forrester2016relating}, we propose the concept of {\textit{partial-skew-orthogonal polynomials (PSOPs)} as a modification of the SOPs}, and then the PSOPs with a variety of special skew-symmetric kernels and weight functions are addressed. By considering appropriate deformations of the weight functions, we derive nine integrable lattices in different dimensions. As a consequence, the tau-functions for these systems are shown to be expressed in terms of Pfaffians and the wave vectors PSOPs. In fact, the tau-functions also admit the representations of multiple integrals. Among these integrable lattices, some of them are known, while the others are novel to the best of our knowledge. In particular, one integrable lattice is related to the partition function of the Bures random matrix ensemble. Besides, we derive a discrete integrable lattice, which can be used to compute certain vector Pad\'e approximants. This yields the first example regarding the connection between integrable lattices and vector Pad\'e approximants, for which Hietarinta, Joshi and Nijhoff pointed out that `` This field remains largely to be explored. '' in the recent monograph \cite[Section 4.4]{hietarinta2016discrete} .

\end{abstract}

\maketitle
\setcounter{tocdepth}{3}
\tableofcontents
\section{Introduction}
The concept of  {\textit {universality}} in the random matrix theory (RMT) has been an incredibly fertile ground for the intimate connections in theoretical physics, statistical mechanics and wireless communication, and many objects in analytic number theory and combinatorics can also be characterized by random matrix models (see e.g. \cite{deift2000random,mehta2004random} and references therein). For example, the distances between scattering resonances, distances between the zeros of Riemann-Zeta function and the patience sorting and longest increasing subsequences are all related to the universality in RMT. From mathematical point of view, the study of universality in energy level correlations of random matrices is connected with the asymptotic behavior of certain families of polynomials. For instances, the spectrum problem of unitary ensembles of random matrices can be transformed into the study of large $N$ asymptotic of standard orthogonal polynomials (OPs), while the spectrum problem of orthogonal/symplectic ensembles of random matrices are related to those of skew orthogonal polynomials (SOPs), which actually originate from RMT \cite{adler2000classical,adler1999pfaff,deift2000orthogonal,dyson1972class,mahoux1991method,mehta2004random}.

The Bures ensemble was identified from the Bures metric, which is a natural choice in measuring the distance of density operators representing states in quantum mechanics \cite{bures1969extension}. It has nowadays attracted many studies in quantum chaos and 2D quantum gravity in RMT (see  \cite{forrester2016relating} and references therein). The Bures ensemble was shown to correspond to a Pfaffian point process by Forrester and Kieburg in \cite{forrester2016relating}.
We observe that the partition functions for the Bures ensemble admit Pfaffian expressions, but they are not exactly included in the frame of SOPs. A natural question arises: what kind of orthogonal polynomials are these partition functions related to? This motivates us to propose a generic notion \textit {partial-skew-orthogonal polynomials (PSOPs)} based on skew-symmetric inner products.  Interestingly, they also arise from the well-posedness of the family of polynomials satisfying some partial-skew orthogonality. And the specific PSOPs related to the Bures ensemble own a four-term recurrence relationship. All of these will be explained in detail in Section \ref{sec:psop}.

As is known, the connection between the theory of OPs and the integrable systems has been increasingly investigated and used by both communities since the early1990s. And OPs play central roles in integrable numerical algorithms in some sense.  In the rest of this paper, our particular interests lies in the investigation on the integrable systems and numerical algorithms associated to PSOPs, before which let's give more introduction on the related topic below.

%This paper is devoted to the study on certain modified skew-orthogonal polynomials and the related integrable lattices, together with some applications to algorithms. 

%Many long-standing problems in the area of orthogonal polynomials have been solved using this connection and several new exciting sides of this connection have been discovered.  \todo{\textcolor{red}{A very recent one is concerned with the relation between orthogonal polynomials and discrete Painlev\'{e} equations using the Riemann-Hilbert formalism for orthogonal polynomials.}}

Regarding OPs and integrable systems, one of the well-known examples is the connection between the ordinary OPs and the Toda lattice, and the Toda lattice is a celebrated completely integrable system. The ordinary OPs appear as wave functions of the Lax pair of the semi-discrete Toda lattice due to a one-parameter deformation of the measure \cite{chu2008linear,nakamura2004special,peherstorfer2007toda}. Besides, the compatibility of spectral transformations of OPs may yield the full-discrete Toda lattice \cite{papageorgiou1995orthogonal,spiridonov1995discrete}.  As a second example,  the semi-discrete Lotka-Volterra lattice (sometimes also called the Kac--van Moerbeke lattice or the Langmuir lattice)  \cite{chang2016multipeakons,chu2008linear,kac1975explicitly,tsujimoto2001discrete} can be obtained as a one-parameter deformation of the measure of symmetric OPs.  The compatibility of spectral transformations of symmetric OPs gives the full-discrete Lotka-Volterra lattice \cite{spiridonov1997discrete}.  For more examples, the readers are invited to consult \cite{adler1995matrix,adler1997string,adler1999generalized,chang2015about,chen2015three} and the references therein.

%The fact that OPs admit Hankel determinant representation leads to the tau-function representation of the Toda lattice in terms of Hankel determinants of the moment sequence. 

Sometimes the full-discrete Toda lattice is also called the qd algorithm \cite{chang2015about,rutishauser1954quotienten}, which can be used to compute eigenvalues of some kind of tridiagonal matrix.    The full-discrete Lotka-Volterra equation can be used as an efficient algorithm to compute singular values of certain tridiagonal matrix \cite{tsujimoto2001discrete}. Here we also mention that one step of the QR algorithm is equivalent to the time evolution of the finite semi-discrete Toda lattice \cite{symes1982qr}. It should be pointed out that, in some sense,  the orthogonality is one underlying ingredient, which plays a central role in the connection between the Toda (and Lotka-Volterra) lattices and the algorithms.  In fact, there exist more connections between the algorithms in numerical analysis and integrable systems. For example, the discrete relativistic Toda molecule equation can be used to design a new Pad\'{e} approximation algorithm \cite{minesaki2001discrete}.  Wynn's celebrated $\varepsilon$-algorithm \cite{wynn1956device} is nothing but the fully discrete potential KdV equation \cite{nakamura2000applied,papageorgiou1993integrable}.  For more information on this topic, please refer \cite{deift1983,hietarinta2016discrete,nakamura2000applied,nakamura2001algorithms,sun2013extended} etc.. Here we stress that, connections of  vector or matrix Pad\'{e} approximants with integrable lattices remain largely to be explored, which was also claimed in the recent monograph \cite[Section4.4]{hietarinta2016discrete} of Hietarinta, Joshi and Nijhoff.

The deformations of the OPs and symmetric OPs lead to the Toda and Lotka-Volterra lattices admiting tau-functions in terms of Hankel determinants of the moment sequence, which can be equivalently rewritten as the forms of multiple integral of Vandermonde determinants.  More deeply, the tau-function of the Toda lattice can be regarded as the partition function of unitary  ensembles in RMT, and the OPs may arise as the corresponding characteristic polynomials over the eigenvalues. In fact, the calculation of statistical properties of the eigenvalues involves many manifestations of integrable system theory and OPs. As a product originating from the study of RMT, the skew-orthogonal polynomials (SOPs) have been extensively studied in the literature ( see, e.g. \cite{adler2000classical,adler1999pfaff,chang2017coupled,deift2000orthogonal,dyson1972class,mahoux1991method,mehta2004random,miki2011discrete} ). Later, the so-called Pfaff lattice \cite{adler1999pfaff} was proposed, whose  tau-functions are Pfaffians and the wave vectors SOPs. And, in \cite{miki2011discrete}, the authors produced full-discrete integrable systems by considering the discrete spectral transformations of SOPs.

In the rest of the paper, we produce nine integrable systems by considering one (or two)-parameter deformations for PSOPs. All of these integrable lattices admit the tau-function representations, in terms of both Pfaffians and also multiple integrals. In particular, a 1+1 dimensional Toda of BKP type is derived by deforming a specific PSOPs, whose $\tau$-function is exactly the time-dependent partition function of the Bures ensemble. Besides, we also obtain two discrete integrable systems, which can be used to accelerate a Pfaffian sequence transformation and a vector Pad\'{e} approximation respectively. It deserves to point out that we provide the first example regarding the connection between the vector Pad\'e approximants and integrable systems.

The paper is organized as follows: The concept of  PSOPs  is introduced in Section \ref{sec:psop} (See Definition \ref{def:psop}). In the subsequent two sections, we derive several integrable lattices  including the two in \cite{chang2017new,chang2017application}  by imposing some special weight functions and deformations on PSOPs.  More exactly, in Section \ref{sec:is}, we obtain seven integrable lattices including (semi or full-discrete) generalized Toda and Lotka-Voterra lattices of BKP type in different dimensions.  In Section \ref{sec:alg}, two discrete integrable lattices are produced, one of which can be used for algorithms for vector Pad\'{e} approximants and the other for convergence acceleration of sequence transformations.  Section \ref{sec:conc} is devoted to conclusion and discussions.

% Skew orthogonal polynomials and their related integrable system have been found for several years. Besides Pfaff lattice, there are just a few known examples related skew-orthogonal polynomials, such as discrete DKP equation \cite{MGT} and coupled modified KdV equation \cite{CHHLTZ}. Thus, we want to explore more equations connected with the polynomials defined by skew inner product. More importantly, we seek for finding some interesting Toda type equations and their related matrix models. For this purpose, firstly, we redefine the  odd-order polynomials in terms of skew inner product.
 \section{Partial-skew-orthogonal polynomials}\label{sec:psop}

  SOPs arises in the theory of random matrices \cite{adler2000classical,dyson1972class,mahoux1991method} and have been extensively investigated. Please refer Appendix \ref{sec:sop} and related references for more details.  In this section, our main aim is to introduce a new concept called {\it partial-skew-orthogonal polynomials} (PSOPs), which are modifications of SOPs. Actually, it may naturally appear as we explain below. 
 
\subsection{Uniqueness of polynomial family with skew-symmetric inner product}
 Let $\langle \cdot, \cdot\rangle$ be a skew-symmetric inner product in the polynomial space over the field of real numbers, more exactly speaking, a bilinear 2-form from $\mathbb{R}(z)\times\mathbb{R}(z)\rightarrow \mathbb{R}$ satisfying the skew symmetric relation:
$$\langle f(z),g(z)\rangle=-\langle g(z),f(z)\rangle.$$
Introduce the bimoment sequence $\{\mu_{i,j}\}_{i,j=0}^\infty$ defined by
\begin{align*}
\mu_{i,j}=\langle z^i,z^j\rangle=-\langle z^j,z^i\rangle.
\end{align*}
 
Under the frame of this skew-symmetric inner product, let's consider the polynomial family $\{P_n(z)\}_{n=0}^{\infty}$
satisfying the usual orthogonality relation:
\begin{align}\label{ortho_even}
\langle P_{n}(z),z^j\rangle=\langle z^j, P_{n}(z)\rangle=0,\qquad 0\leq j\leq n-1.
\end{align}

\subsubsection{Even degree case} 
 As for the monic polynomials of exactly even degree with the ansatz
$$
P_{2n}(z)=z^{2n}+\sum_{i=1}^{2n-1}a_{n,i}z^{i},
$$
the orthogonality relations
$$
\langle P_{2n}(z),z^j\rangle=\langle z^j, P_{2n}(z)\rangle=0,\qquad 0\leq j\leq 2n-1
$$
lead to the linear system
$$
\sum_{i=0}^{2n-1}a_{n,i}\mu_{j,i}+\mu_{j,2n}=0,\qquad 0\leq j\leq 2n-1.
$$
If the determinant $\det\left((\mu_{i,j})_{i,j=0}^{2n-1}\right)$ of coefficient matrix is nonzero, then this linear system is uniquely solved by using Cramer's rule, leading to
\begin{align*}
 P_{2n}(z)=\frac{1}{\det\left(\mu_{i,j}\right)_{0\leq i,j\leq 2n-1}}
 \left|
 \begin{array}{cccc}
 \mu_{0,0}& \mu_{0,1}&\cdots& \mu_{0,2n}\\
  \mu_{1,0}& \mu_{1,1}&\cdots& \mu_{1,2n}\\
  \vdots&\vdots&\ddots&\vdots\vdots\\
   \mu_{2n-1,0}& \mu_{2n-1,1}&\cdots& \mu_{2n-1,2n}\\
   1&z&\cdots&z^{2n}
 \end{array}
 \right|.
\end{align*}
Actually, the skew-symmetric property of the moment matrix $(\mu_{i,j})_{i,j=0}^{k}$ promotes one to give the representations in terms of Pfaffians \footnote{Please see more information of Pfaffians in Appendix \ref{app_pf}. And note that we will use the notations due to Hirota \cite{hirota2004direct}.} , which was noticed by Adler, Horozov and van Moerbeke in \cite{adler1999pfaff} for the SOPs, whose even case is the same as here.  
 Based on two facts  \eqref{det_pf1}
and \eqref{det_pf_even} between the determinants and Pfaffians, one can obtain that
 the monic polynomials $\{P_{2n}(z)\}_{n=0}^{\infty}$ have the following explicit form in terms of Pfaffians 
\begin{align*}
&P_{2n}(z)=\frac{1}{\tau_{2n}}Pf(0,1,\cdots,2n-1,2n,z),
\end{align*}
where 
the Pfaffian entries are defined as 
$$Pf(i,j)=\mu_{i,j},\qquad Pf(i,z)=z^i$$ 
and
$$\tau_{2n}\triangleq Pf(0,1,\cdots,2n-1)\neq0.$$

 \subsubsection{Odd degree case}For the the monic polynomials of odd degree with the ansatz 
 $$
P_{2n+1}(z)=z^{2n+1}+\sum_{i=1}^{2n}b_{n,i}z^{i},
$$ if one still employ the orthogonality relations
$$
\langle P_{2n+1}(z),z^j\rangle=\langle z^j, P_{2n+1}(z)\rangle=0,\qquad 0\leq j\leq 2n,
$$ 
an ill-defined linear system
$$
\sum_{i=0}^{2n}b_{n,i}\mu_{j,i}+\mu_{j,2n+1}=0,\qquad 0\leq j\leq 2n
$$
 arises since the coefficient matrix is a skew symmetric matrix of odd order, which is singular. It shows that it is not a good way for such orthogonality definition for odd case, which motivates one to search for other definitions. 
 
Generally,  for any monic polynomial  family of odd order, there hold
\begin{align}\label{gene_inn}
\langle P_{2n+1}(z),z^j\rangle=\alpha_{j,2n+1},\qquad 0\leq j\leq 2n+1.
\end{align}
Here one question is naturally raised:  can we choose some appropriate $\{\alpha_{j,2n+1}\}_{0\leq j\leq 2n+1}$ so that the monic polynomial  family might be uniquely determined?  Let's give some specific assignments below.

The inner product relations  \eqref{gene_inn} yield
\begin{align}
\sum_{i=0}^{2n}b_{n,i}\mu_{j,i}+\mu_{j,2n+1}=-\alpha_{j,2n+1},\qquad 0\leq j\leq 2n+1.
\end{align}
Let $A_{k,l} $ be a matrix defined by 
$$ A_{k,l}=\left(\mu_{i,j}\right)_{0\leq i\leq k-1, 0\leq j\leq l-1}$$
and 
$B_{k,l}$ be a matrix of $A_{k,l-1}$ added an extra column, namely
$$
B_{k,l}=
\left(
\begin{array}{ccc}
A_{k,l-1}&\vline&
\begin{array}{c}
\mu_{0,l-1}+\alpha_{0,l-1}\\
\vdots\\
\mu_{k-1,l-1}+\alpha_{k-1,l-1}
\end{array}
\end{array}
\right).$$
To intend the above linear system has a unique nonzero solution is equivalent to need
$$\rank(A_{2n+2,2n+1})=\rank(B_{2n+2,2n+2})=2n+1.$$

One one hand, it follows from the assumption on the even case that $A_{2n+2,2n+2}$ is of full rank, which implies that $A_{2n+2,2n+1}$, i.e. the coefficient matrix, has rank $2n+1$.

On the other hand, if we can find appropriate $\alpha_{j,2n+1}$ satisfying $\det(B_{2n+2,2n+2})=0$, then the rank of $B_{2n+2,2n+2}$ is 2n+1. And the monic polynomial  family of odd degree will be determined uniquely.

The specific choice 
\begin{align*}
&\alpha_{2n,2n+1}=\det(A_{2n+2,2n+2}){\Big/}\det\left(
\begin{array}{ccc}
A_{2n+1,2n}&\vline&
\begin{array}{c}
\mu_{0,2n+1}\\
\vdots\\
\mu_{2n,2n+1}
\end{array}
\end{array}
\right),\\
 & \alpha_{j,2n+1}=0,\qquad j=0,1\cdots,2n-1,2n+1
\end{align*}
leads to $\det(B_{2n+2,2n+2})=0$. Consequently, this choice uniquely determines a class of monic polynomials of odd degree, that is, \begin{align}
 P_{2n+1}(z)=\frac{1}{\det\left(\mu_{i,j}\right)_{0\leq i,j\leq 2n-1}}
 \left|
 \begin{array}{ccccc}
 \mu_{0,0}& \mu_{0,1}&\cdots& \mu_{0,2n-1}&\mu_{0,2n+1}\\
  \mu_{1,0}& \mu_{1,1}&\cdots& \mu_{1,2n-1}&\mu_{1,2n+1}\\
  \vdots&\vdots&\ddots&\vdots&\vdots\\
     \mu_{2n-1,0}& \mu_{2n-1,1}&\cdots& \mu_{2n-1,2n-1}&\mu_{2n-1,2n+1}\\
      1&z&\cdots&z^{2n-1}&z^{2n+1}
 \end{array}
 \right|,\label{sop_odd_det}
\end{align}
which are nothing but the case of SOPs.

In the following, let us consider another possibility. Assume that the $\alpha_{j,2n+1}$ are restricted by
\begin{align}\label{ass_alpha}
\alpha_{j,2n+1}=-\beta_j\frac{\det(A_{2n+2,2n+2})}{\det(C_{2n+2,2n+2})},
\end{align}
where $\beta_j$ are some constants and the nonsingular matrix $C_{k,l}$ are defined by
$$
C_{k,l}=
\left(
\begin{array}{ccc}
A_{k,l-1}&\vline&
\begin{array}{c}
\beta_0\\
\vdots\\
\beta_{k-1}
\end{array}
\end{array}
\right).$$
 By multi-linearity, we have
\begin{align*}
\det(B_{2n+2,2n+2})=\det(A_{2n+2,2n+2})-\frac{\det(A_{2n+2,2n+2})}{\det(C_{2n+2,2n+2})}\det(C_{2n+2,2n+2})=0.
\end{align*}
Thus this specific choice also uniquely determines the odd polynomial family satisfying \eqref{gene_inn}. Although they can not be explicitly solved by use of Cramer's rule, it is not hard to guess that 
%\begin{align*}
%P_{2n+1}(z)=\frac{1}{\det\left(\mu_{i,j}\right)_{0\leq i,j\leq 2n-1}}
% \left|
% \begin{array}{ccccc}
% \mu_{0,0}& \mu_{0,1}&\cdots& \mu_{0,2n}&\mu_{0,2n+1}\\
%  \mu_{1,0}& \mu_{1,1}&\cdots& \mu_{1,2n}&\mu_{1,2n+1}\\
%  \vdots&\vdots&\ddots&\vdots&\vdots\\
%     \mu_{2n,0}& \mu_{2n,1}&\cdots& \mu_{2n,2n}&\mu_{2n-1,2n+1}\\
%      1&z&\cdots&z^{2n}&z^{2n+1}
% \end{array}
% \right|,
%\end{align*}
they can be written as 
\begin{align*}
P_{2n+1}(z)=\frac{1}{\det\left(C_{2n+2,2n+2}\right)}
 \left|
 \begin{array}{ccccc}
 \mu_{0,0}& \mu_{0,1}&\cdots& \mu_{0,2n+1}&-\beta_0\\
  \mu_{1,0}& \mu_{1,1}&\cdots& \mu_{1,2n+1}&-\beta_1\\
  \vdots&\vdots&\ddots&\vdots&\vdots\\
     \mu_{2n+1,0}& \mu_{2n+1,1}&\cdots& \mu_{2n+1,2n+1}&-\beta_{2n+1}\\
      1&z&\cdots&z^{2n+1}&0
 \end{array}
 \right|,
\end{align*}
if $\det\left(C_{2n+2,2n+2}\right)\neq0$. One could confirm the inner product relations \eqref{gene_inn} with $\alpha_{j,2n+1}$ in \eqref{ass_alpha} are satisfied. 

It is noted that the expressions above can also be written in terms of Pfaffians. By using
\eqref{det_pf1}
and \eqref{det_pf_odd}, one can obtain that
 the monic polynomials $\{P_{2n}(z)\}_{n=0}^{\infty}$ have the following explicit form in terms of Pfaffians 
\begin{align*}
P_{2n+1}(z)=\frac{1}{\tau_{2n+1}}Pf(d_0,0,1,\cdots,2n,2n+1,z),
\end{align*}
where 
the Pfaffian entries are defined as 
\begin{align*}
&Pf(i,j)=\mu_{i,j}\triangleq \langle z^i,z^j\rangle,\qquad Pf(d_0,i)=\beta_i,\\
&Pf(i,z)=z^i, \qquad\qquad\qquad\quad Pf(d_0,z)=0,
\end{align*}
and
$$\tau_{2n+1} \triangleq Pf(d_0,0,1,\cdots,2n)\neq0.$$

\begin{remark}
The specific assignment \eqref{ass_alpha} seems strange, but it does come from something interesting. In fact, we are motivated by the average of characteristic polynomials of Bures ensemble in \cite{forrester2016relating}. The relation will be explained in the subsection \ref{subsec:psop_bures} rather than here.
\end{remark}

\subsection{PSOPs with Pfaffian expressions}

Based on the argument in the above subsection, we are finally inspired to give the following definition. 
 \begin{define}\label{def:psop}
Let $\langle \cdot, \cdot\rangle$ be a  skew-symmetric inner product in the polynomial space over the field of real numbers, more exactly speaking, a bilinear 2-form from $\mathbb{R}(z)\times\mathbb{R}(z)\rightarrow \mathbb{R}$ satisfying the skew symmetric relation:
$$\langle f(z),g(z)\rangle=-\langle g(z),f(z)\rangle.$$
A family of monic polynomials $\{P_n(z)\}_{n=0}^{\infty}$ are called monic partial-skew-orthogonal polynomials (PSOPs)  with respect to the  skew-symmetric inner product $\langle \cdot, \cdot\rangle$ if they satisfy the following relations:
\begin{subequations}\label{pskew_inner}
\begin{align}
&\langle P_{2n}(z),z^m\rangle=\zeta_n\delta_{2n+1,m},\qquad 0\leq j\leq 2n-1
,\label{pskew_inner1}\\
&\langle P_{2n+1}(z),z^m\rangle=\beta_m\gamma_n,\quad \qquad 0\leq j\leq 2n+1
,\label{pskew_inner2}
%&\langle P_{2n+1}(z), z^i\rangle=\beta_{n,i}, 0\leq i\leq 2n+1,\label{pskew_inner2}
\end{align}
\end{subequations}
for some appropriate nonzero numbers $\beta_m,\gamma_n,\zeta_n$.
 \end{define}

%where the fact is used
%$$\det(C_{2n+2,2n+2})={Pf(0,1,\cdots,2n+1)}{Pf(d_0,0,1,\cdots,2n)}$$
%by \eqref{det_pf_odd}. Thus the proof is completed.

As is discussed in the last subsection, we are mainly interested in the monic PSOP sequence $\{P_n(z)\}_{n=0}^\infty$ written as:
\begin{subequations}\label{exp:PSOP}
\begin{align} 
&P_{2n}(z)=\frac{1}{\tau_{2n}}Pf(0,1,\cdots,2n-1,2n,z), \label{exp:PSOP_even}\\
&P_{2n+1}(z)=\frac{1}{\tau_{2n+1}}Pf(d_0,0,1,\cdots,2n,2n+1,z),\label{exp:PSOP_odd}
\end{align}
\end{subequations}
with the assumption
$$
\tau_{2n}\triangleq Pf(0,1,\cdots,2n-1)\neq 0,\qquad \tau_{2n+1} \triangleq Pf(d_0,0,1,\cdots,2n)\neq0,
$$
where the Pfaffian entries are defined as
\begin{align*}
&Pf(i,j)=\mu_{i,j}\triangleq \langle z^i,z^j\rangle,\qquad Pf(d_0,i)=\beta_i,\\
&Pf(i,z)=z^i,\qquad\qquad\qquad\quad Pf(d_0,z)=0.
\end{align*}

The ``partial-skew-orthogonality'' property for this class of PSOPs is listed in the following theorem.
\begin{theorem} As for the polynomial sequence \eqref{exp:PSOP},  there hold
\begin{align*}
&\langle P_{2n}(z), z^m\rangle=\frac{\tau_{2n+2}}{\tau_{2n}}\delta_{2n+1,m},\qquad 0\leq m\leq 2n+1,\\
& \langle P_{2n+1}(z), z^m\rangle=-\frac{\tau_{2n+2}}{\tau_{2n+1}}\beta_m,\qquad\ 0\leq m\leq 2n+1.
\end{align*}

\end{theorem}
 \begin{proof} 
Actually, the conclusion for even-odd polynomials immediately follows from those for SOPs. Here we would like to give a direct proof as follows: 
   \begin{align*}
\langle P_{2n}(z), z^{m}\rangle&=\frac{1}{\tau_{2n}}\langle Pf(0,1,\cdots,2n-1,2n,z), z^m\rangle\\
  &=\frac{1}{\tau_{2n}} \sum_{j=0}^{2n}(-1)^jPf(0,1,\cdots,\hat j,\cdots,2n) \langle z^j, z^m\rangle\\
  &=\frac{1}{\tau_{2n}} \sum_{j=0}^{2n}(-1)^jPf(0,1,\cdots,\hat j,\cdots,2n)Pf(j,m)\\
  &=\frac{1}{\tau_{2n}}Pf(0,1,\cdots,2n,m)
    \end{align*}
    leading to
    $$
     \langle P_{2n}(z), z^{m}\rangle=0, \quad 0\leq m\leq 2n,\qquad \qquad \langle P_{2n}(z), z^{2n+1}\rangle=\frac{\tau_{2n+2}}{\tau_{2n}},
    $$
    which can be equivalently written as 
    \begin{align*}
&\langle P_{2n}(z), P_{2m-1}(z)\rangle=\langle P_{2n}(z),  P_{2m}(z)\rangle=0,\qquad 0\leq m\leq n,\\
&\langle P_{2n}(z), P_{2n+1}(z)\rangle=\frac{\tau_{2n+2}}{\tau_{2n}}.
    \end{align*}
    
    Regarding $\langle P_{2n+1}(z), z^{m}\rangle,\ 0\leq m\leq 2n+1$, we have
     \begin{align*}
\langle P_{2n+1}(z), z^{m}\rangle&=\frac{1}{\tau_{2n+1}}\langle Pf(d_0,0,1,\cdots,2n,2n+1,z), z^m\rangle\\
  &=\frac{1}{\tau_{2n+1}} \sum_{j=0}^{2n+1}(-1)^{j+1}Pf(d_0,0,1,\cdots,\hat j,\cdots,2n,2n+1) \langle z^j, z^m\rangle\\
  &=\frac{1}{\tau_{2n+1}} \sum_{j=0}^{2n+1}(-1)^{j+1}Pf(d_0,0,1,\cdots,\hat j,\cdots,2n,2n+1)Pf(j,m)\\
  &=-\frac{1}{\tau_{2n+1}} Pf(0,1,\cdots,2n,2n+1)Pf(d_0,m)=-\frac{\tau_{2n+2}}{\tau_{2n+1}}\beta_m,
    \end{align*}
 where we have used the expansion formula for $Pf(d_0,0,1,\cdots,2n,2n+1,m)$ (which equals to zero for $0\leq m\leq 2n+1$) in the last second equality. 
% \begin{align*}
%  &Pf(0,1,\cdots,2n,2n+1)Pf(d_0,m)+\sum_{j=0}^{2n+1}(-1)^{j+1}Pf(d_0,0,1,\cdots,\hat j,\cdots,2n,2n+1)Pf(j,m)\\
%  =&=0
% \end{align*}

 \end{proof}

At the end of this subsection, we restrict ourselves into the specific  skew-symmetric inner products 
$$
\langle f(z),g(z)\rangle= \iint_{\Gamma^2}f(x)g(y)\omega(x,y)\rho(x)\rho(y)dxdy,
$$
and the single moments
$$\beta_i=\int_\Gamma x^i\rho(x)dx$$
where $\omega(x,y)=-\omega(y,x)$ and the $\Gamma$ is an appropriate integral range and $\rho(x)dx$ is taken as a positive weight ensuring the validity and finiteness of the integral. Usually, if $\Gamma$ is replaced by $\mathbb R$ or $\mathbb R_+$, then $\rho(x)$ is required to decay fast enough at $\infty$. We will provide a sufficient condition to ensure the validity of the definition of the PSOPs in \eqref{exp:PSOP},  in other words,  the nonzero property of $\tau_n$ therein.

\begin{theorem}\label{th:exis}
Assume that 
$$Pf(i,j)=\mu_{i,j}= \iint_{\Gamma^2}x^iy^j\omega(x,y)\rho(x)\rho(y)dxdy,\qquad Pf(d_0,i)= \beta_i=\int_\Gamma x^i\rho(x)dx,$$
and $\Gamma$ is an ordered measure space. Let any $W(x_1,\cdots,x_{N})$ be a skew symmetric matrix defined by
$$ 
 W(x_1,\cdots,x_{N})=\left\{
 \begin{array}{ll}
 \left(\begin{array}{c}
\omega(x_i,x_j)
\end{array}
\right)_{(2n)\times (2n)}, & N=2n,\\
\ \\
\left(
\begin{array}{ccc}
0&\vline&
\begin{array}{ccc}
1&\cdots&1
\end{array}
\\
\hline
\begin{array}{c}
-1\\
\vdots\\
-1
\end{array}
&\vline&\omega(x_i,x_j)
\end{array}
\right)_{(2n+2)\times (2n+2)}, &N=2n+1.
 \end{array} 
 \right. 
$$ 
Then PSOPs in \eqref{exp:PSOP} are well-defined if the Pfaffians $Pf(W(x_1,\cdots,x_{N}))$ are positive (or negative) a.e. for all $x_1<\cdots<x_{N}$ over the region $\Gamma$. 

%By contrast, the odd-degree PSOPs in \eqref{exp:PSOP_odd} are well-defined if the Pfaffians $Pf(W(x_1,\cdots,x_{2n+1}))$ are positive (or negative) a.e. for all $x_1<\cdots<x_{2n+1}$ over the region $\Gamma$. Here $W(x_1,\cdots,x_{2n+1}))$ is a skew-symmetric matrix of order $2n+2$:
%$$ W(x_1,\cdots,x_{2n+1})=
%\left(\begin{array}{cccc}
%0&1&\cdots&1\\
%-1&\\
%\vdots&&s(x_i,x_j)\\
%-1
%\end{array}
%\right)_{(2n+2)\times (2n+2)},$$ 
%%adding a row element $1$ and column on the $W(x_1,\cdots,x_{2n+1}))$.
\end{theorem}
\begin{proof}
By employing the de Bruijn's equality \eqref{id_deBr1},  we have
\begin{align*}
Pf(0,1,\ldots, 2n-1)&=\idotsint\limits_{x_1<\cdots<x_{2n}\in \Gamma} Pf(W(x_1,\cdots,x_{2n}))\det\left(x_j^{i-1}\right)_{1\leq i,j\leq 2n}\prod_{j=1}^{2n} \rho(x_i)dx_1\cdots dx_{2n}\\ 
&=\idotsint\limits_{x_1<\cdots<x_{2n}\in \Gamma} Pf(W(x_1,\cdots,x_{2n}))\prod_{1\leq i<j\leq 2n}(x_j-x_i)\prod_{j=1}^{2n} \rho(x_i)dx_1\cdots dx_{2n},
\end{align*}
which conclude the validity of the even-degree PSOPs. Similarly,  the odd case is just a consequence of applying the de Bruijn's equality \eqref{id_deBr2}.
\end{proof}

\subsection{Relation with characteristic polynomials of Bures ensemble}\label{subsec:psop_bures}
 The idea of the definition of PSOPs  (Def. \ref{def:psop}) partially comes from the characteristic polynomials of Bures ensemble appeared  in \cite{forrester2016relating}. In the following, we demonstrate their relations.
 
Consider the specified PSOPs \eqref{exp:PSOP} equipped with 
$$\mu_{i,j}= \iint_{\mathbb{R}^2}x^iy^j\omega(x,y)\frac{x-y}{x+y}\rho(x)\rho(y)dxdy,\quad \beta_i=\int_\mathbb{R} x^i\rho(x)dx.$$
Applying the de Brujin's formulae \eqref{id_deBr1}-\eqref{id_deBr2} and the Schur's Pfaffian identities \eqref{id_schur}-\eqref{id_schur_odd}, it is not difficult to see that 
\begin{align*}
\tau_n=(-1)^{\lfloor\frac{n}{2}\rfloor}\idotsint\limits_{0<x_1<\cdots<x_n<\infty}\prod_{1\leq i<j\leq n}\frac{(x_j-x_i)^2}{(x_i+x_j)}\prod_{1\leq j\leq n}\rho(x_j)d x_1\cdots dx_n\neq 0,
\end{align*}
where $\lfloor a\rfloor$ denotes the greatest integer less than or equal to $a$. Thus the specified PSOPs are well defined.

In fact, the specified PSOPs may arise as the average of characteristic polynomials of Bures ensemble \cite{forrester2016relating}.  Observe that there holds the equality 
\begin{align*}
\det\left(\begin{array}{cccccc}
1&x_1&\cdots&\widehat{x_1^{n-k}}&\cdots&x_1^n\\
1&x_2&\cdots&\widehat{x_2^{n-k}}&\cdots&x_2^n\\
\vdots&\vdots&\ddots&\vdots&\ddots&\vdots\\
1&x_n&\cdots&\widehat{x_n^{n-k}}&\cdots&x_n^n
\end{array}
\right)=\sum_\sigma x_{i_1}x_{i_2}\cdots x_{i_k}\prod_{1\leq i<j\leq n}(x_j-x_i),
\end{align*}
where $\sum_{\sigma}$ means the summation from all possible combination satisfying the condition $1\leq i_1<\cdots<i_k\leq n$ and \ $\widehat{\cdot}$ implies the corresponding column is removed. Then, applying the de Brujin's formulae \eqref{id_deBr1} and  \eqref{id_deBr2}, 
one can obtain that the specified PSOPs own a matrix integral form in terms of 
\begin{align*}
{P}_{n}(z)=\frac{(-1)^{\lfloor\frac{n}{2}\rfloor}}{\tau_n(t)}\ \ \ \ \idotsint\limits_{0<x_1<\cdots<x_n<\infty}\prod_{1\leq i<j\leq n}\frac{(x_j-x_i)^2}{(x_i+x_j)}{\prod_{1\leq i<j\leq n}(z-x_j)}\rho(x_j)d x_1\cdots dx_n.
\end{align*}
This is nothing but the average of the characteristic polynomials of Bures ensemble, i.e.
\begin{align*}
{P}_{n}(z)=\langle \det(zI-J)\rangle=\langle \prod_{j=1}^n (z-x_j)\rangle.
\end{align*}

Furthermore, it is noted that, the specified PSOPs can be expressed by Meijer-G functions \cite{forrester2016relating} if the Laguerre weight function $\rho(x)=e^{-x}x^a$ is taken into account. And, it also deserves to remark that this specific PSOPs admit a four-term recurrence relationship, which will be presented Theorem \ref{psop_4term}.

  \section{PSOPs and related integrable lattices}\label{sec:is}
  In the previous section, a class of PSOPs are introduced based on the formal moments $\mu_{i,j}$ and $\beta_i$. Now, we will assign some more specific $\mu_{i,j}$ and $\beta_i$ and impose different $t,s$-deformation on the PSOPs to produce some integrable lattices together with their tau-functions. 
  
Again, it is noted that, for our convenience, we mainly consider some specific well-defined $\mu_{i,j}$ and $\beta_i$ on a positive weight $\rho(x)dx$  on $\mathbb{R}$ (or $\mathbb{R_+}$) with $\rho(x)$ decaying fast enough at $\infty$ and a skew integral kernel $\omega(x,y)$ enjoying the property $\omega(x,y)=-\omega(y,x)$. Roughly speaking,  $\mu_{i,j}$ and $\beta_i$ are taken as the following form:
  $$\mu_{i,j}= \iint_{\mathbb{R}^2}x^iy^j\omega(x,y)\rho(x)\rho(y)dxdy,\qquad  \beta_i=\int_\mathbb{R} x^i\rho(x)dx,$$
  which is beneficial for producing integrable lattices.
 Note that the sufficient condition for the validity of the corresponding PSOPs is given in Theorem \ref{th:exis}.

%Consider the weight $\rho(x)dx$  on $\mathbb{R}$ with $\rho(x)$ decaying fast enough at $\infty$ so that the moments $\mu_{i,j}$ and $\beta_j$.

%, depending on $t = (t_1, t_2,\cdots)^\top\in \mathbb{C}^\infty $
\subsection{Generalized Lotka-Volterra lattices} \label{subsec:glv}
This subsection is devoted to the derivations of a semi-discrete generalized Lotka-Volterra lattice and a full-discrete counterpart. Note that the skew-symmetric kernel $\omega(x,y)$ is not assigned. As we will see soon in the subsection \ref{subsec:btoda}, when the skew-symmetric kernel is assigned as some special function, some different lattices will be derived. 
\subsubsection{A semi-discrete gLV lattice in $1+2$ dimension}
%\subsubsection{A generalized semi-discrete Lotka-Volterra lattice in $1+2$ dimension}
Consider the deformation on the weight $\rho(x;t,k)dx=x^k\exp(tx)\rho(x)dx$  leading to the $k$-index moments  depending on $t \in \mathbb{R}$:
\begin{align*}
&\mu_{i,j}^k(t)=\iint_{\mathbb{R}^2}x^{k+i}y^{k+j}\omega(x,y)\exp\left(t(x+y)\right)\rho(x)\rho(y)dxdy,\\
&\beta_i^k(t)=\int_{\mathbb{R}}x^{k+i}\exp\left (tx\right)\rho(x)dx.
\end{align*}
And, define a sequence of monic PSOPs ($\{P_n^k(z;t)\}_{n=0}^\infty, k\in \mathbb{N}$) depending on $t$:
\begin{subequations}\label{psop_sem_lv}
\begin{align} 
&P_{2n}^k(z;t)=\frac{1}{\tau_{2n}^k}Pf(0,1,\cdots,2n-1,2n,z), \label{psop_sem_lv_even}\\
&P_{2n+1}^k(z;t)=\frac{1}{\tau_{2n+1}^k}Pf(d_0,0,1,\cdots,2n,2n+1,z),\label{psop_sem_lv_odd}
\end{align}
\end{subequations}
where
\begin{align}
\tau_{2n}^k(t)\triangleq Pf(0,1,\cdots,2n-1)^k\neq 0,\qquad \tau_{2n+1}^k(t) \triangleq Pf(d_0,0,1,\cdots,2n)^k\neq0, \label{psop_sem_lv_tau}
\end{align}
with the Pfaffian entries
\begin{align*}
&Pf(i,j)^k=\mu_{i,j}^k,\qquad\qquad Pf(d_0,i)^k=\beta_i^k,\\
&Pf(i,z)^k=z^i,\qquad\qquad\ \  Pf(d_0,z)^k=0.
\end{align*}
 
 First, we observe some interrelation among $\tau_n^k$ and claim that the evolution of $\tau_n^k$ can also be explicitly expressed in terms of Pfaffians.
 \begin{lemma}
 There hold the following relations for the $\tau_n^k$ defined in \eqref{psop_sem_lv_tau}:
  \begin{align}
  \tau_{2n}^{k+1}=Pf(1,2,\cdots,2n)^k,\qquad \tau_{2n+1}^{k+1}= Pf(d_0,1,2,\cdots,2n+1)^k, \label{psop_sem_lv_tau_rela}
  \end{align}
  subsequently leading to
  \begin{align}
z{\tau_{2n}^{k+1}}P_{2n}^{k+1}=Pf(1,\cdots,2n+1,z)^k,\quad z{\tau_{2n+1}^{k+1}}P_{2n+1}^{k+1}=Pf(d_0,1,\cdots,2n+2,z)^k.\label{psop_sem_lv_rel1}
\end{align}
Besides, the $\tau_n^k$  evolves as follows:
\begin{align}
\frac{d}{d{t}}\tau_{2n}^{k}=Pf(0,1,\cdots,2n-2,2n)^k,\qquad \frac{d}{dt} \tau_{2n+1}^{k}= Pf(d_0,0,1,\cdots,2n-1,2n+1)^k.\label{psop_sem_lv_tau_evo}
\end{align}
 \end{lemma}
 \begin{proof} 
The conclusion \eqref{psop_sem_lv_tau_rela} and \eqref{psop_sem_lv_rel1} are obtained by using the definition of the Pfaffian and 
noticing that 
 $$\mu_{i,j}^{k+1}=\mu_{i+1,j+1}^k,\qquad \beta_i^{k+1}=\beta_{i+1}^k.$$
 
 It is not hard to see that 
 $$
 \frac{d}{d{t}} \mu_{i,j}^k=\mu_{i+1,j}^k+\mu_{i,j+1}^k,\qquad  \frac{d}{d{t}} \beta_i^k=\beta_{i+1}^k,
 $$
from which, we can conclude the evolution \eqref {psop_sem_lv_tau_evo} by employing the derivative formula \eqref{der1} and \eqref{der1_odd}.
 \end{proof}
 
 This lemma leads to an obvious corollary on the representations of the PSOPs in \eqref{psop_sem_lv}, that is, the coefficient of the second highest degree and the constant term can be written in terms of $\tau_n^k$.
\begin{coro}
The PSOPs in \eqref{psop_sem_lv} admit the following expressions:
\begin{align}
P_n^k(z;t)=z^n-(\frac{d}{d t}\log\tau_n^k )z^{n-1}+\cdots+(-1)^n\frac{\tau_n^{k+1}}{\tau_n^k}. \label{psop_sem_lv_exp_tau}
\end{align}
\end{coro}
 
 Usually, the transformation between adjacent families of the polynomials $\{P_{n}^{k+1}\}_{n=0}^\infty$ and $\{P_{n}^k\}_{n=0}^\infty$ also plays and important role in the construction of integrable systems from OPs. Actually, as for the PSOPs \eqref{psop_sem_lv}, we have 
\begin{coro}
There holds the following relation between adjacent families of the polynomials $\{P_{n}^{k+1}\}_{n=0}^\infty$ and $\{P_{n}^k\}_{n=0}^\infty$:
\begin{align}
z\left(P_{n}^{k+1}+\frac{\tau_{n-1}^{k+1}\tau_{n+2}^k}{\tau_{n}^{k+1}\tau_{n+1}^k}P_{n-1}^{k+1}\right)=P_{n+1}^k+\frac{\tau_{n+1}^{k+1}\tau_{n}^k}{\tau_{n}^{k+1}\tau_{n+1}^k}P_{n}^k. \label{psop_sem_lv_rel_adj}
\end{align}
\end{coro}
\begin{proof}
 Replace the terms including $(k+1)$-index in  \eqref{psop_sem_lv_rel_adj} by using \eqref{psop_sem_lv_tau_rela} and \eqref{psop_sem_lv_rel1},  and then we can see that it suffices to prove  
 \begin{align*}
&Pf(d_0,0,1,\cdots,2n,2n+1,z)^kPf(1,\cdots,2n)^k=Pf(d_0,0,1,\cdots,2n)^kPf(1,\cdots,2n,2n+1,z)^k\\
&\ \ -Pf(d_0,1,\cdots,2n,2n+1)^kPf(0,1,\cdots,2n,z)^k+Pf(d_0,1,\cdots,2n,z)^kPf(0,1,\cdots,2n,2n+1)^k,\\
&Pf(d_0,0,1,\cdots,2n-1,2n)^kPf(1,\cdots,2n-1,z)^k=Pf(d_0,0,1,\cdots,2n-1,z)^kPf(1,\cdots,2n-1,2n)^k\\
&\ \ -Pf(d_0,1,\cdots,2n-1,2n,z)^kPf(0,1,\cdots,2n-1)^k+Pf(0,1,\cdots,2n-1,2n,z)^kPf(d_0,1,\cdots,2n-1)^k,
\end{align*}
which are nothing but Pfaffian identities in \eqref{pf1} and \eqref{pf2}. Thus the proof is completed.
\end{proof}

The compatibility of the relations \eqref{psop_sem_lv_exp_tau} and \eqref{psop_sem_lv_rel_adj} can yield an ODE system. To this end, we notice that the coefficients of the highest degree in $z$ on the two sides of \eqref{psop_sem_lv_rel_adj} are both 1, while the coefficients of the second highest degree are not trivial.
Actually, substituting \eqref{psop_sem_lv_exp_tau} into \eqref{psop_sem_lv_rel_adj} and comparing the coefficients of $z^{n}$ on two sides, we obtain
\begin{align}\label{eq:bi_glv}
\tau_{n+1}^k\frac{d}{dt} \tau_{n}^{k+1}-\tau_{n}^{k+1}\frac{d}{dt} \tau_{n+1}^k =\tau_{n+2}^{k}\tau_{n-1}^{k+1}-\tau_{n}^k\tau_{n+1}^{k+1}.
\end{align}

If we introduce the variables
\begin{align*}
r_n^k=\frac{\tau_n^{k+1}}{\tau_n^k},\qquad v_n^k=\frac{\tau_{n+1}^k}{\tau_{n}^{k+1}}, 
\end{align*}
then $r_n^k$ and $v_n^k$ satisfy
\begin{align}
\frac{d}{dt}v_n^k=&\frac{v_n^kr_{n+1}^k}{v_{n-1}^kr_{n}^k}(v_{n-1}^k-v_{n+1}^k),\qquad v_n^kr_{n}^k=r_{n+1}^{k-1}v_n^{k-1}, \quad n,k=0,1,2,\cdots, \label{eq:21glv}
\end{align}
which we shall call ``generalized Lotka-Volterra lattice in 1+2 dimension'' (1+2 gLV lattice). It is noted that the bilinear form \eqref{eq:bi_glv} reduces to the Lotka-Volterra lattice when we remove the index $k$, which is the origin that we call \eqref{eq:21glv} 1+2 gLV lattice. And ``1+2'' means one continuous and two discrete variables.

The above derivation implies the following theorem, which concludes the tau-function representation of the equation  \eqref{eq:21glv}.
\begin{theorem}
The 1+2 gLV lattice \eqref{eq:21glv} admits the following tau-function representation
\begin{align*}
r_n^k=\frac{\tau_n^{k+1}}{\tau_n^k},\qquad v_n^k=\frac{\tau_{n+1}^k}{\tau_{n}^{k+1}}, 
\end{align*}
which enjoys the explicit form in terms of Pfaffians:
$$
\tau_{2n}^k(t)=Pf(0,1,\cdots,2n-1)^k,\qquad \tau_{2n+1}^k(t) = Pf(d_0,0,1,\cdots,2n)^k,
$$
with the Pfaffian entries satisfying
\begin{align*}
&Pf(i,j)^k=\mu_{i,j}^k,\qquad\qquad\qquad Pf(d_0,i)^k=\beta_i^k,\\
&\mu_{i,j}^{k+1}=\mu_{i+1,j+1}^k,\qquad\qquad\qquad\beta_i^{k+1}=\beta_{i+1}^k,\\
& \frac{d}{d{t}} \mu_{i,j}^k=\mu_{i+1,j}^k+\mu_{i,j+1}^k,\qquad\ \ \frac{d}{d{t}} \beta_i^k=\beta_{i+1}^k.
\end{align*}
%or an equivalent integral form
%$$
%\tau_{n}^k(t)\idotsint\limits_{0<x_1<\cdots<x_n<\infty}\prod_{1\leq i<j\leq n}\frac{(x_j-x_i)^2}{(x_i+x_j)}{\prod_{1\leq j\leq n}}e^{tx_j}\rho(x_j)d x_1\cdots dx_n.
%$$
\end{theorem}
\begin{remark}
It is not hard to see that the de Bruijn's formulae \eqref{id_deBr1}-\eqref{id_deBr2} imply that the tau-function owns an equivalent integral formula, that is,
\begin{align*}
\tau_{n}^k(t)
&=\idotsint\limits_{-\infty<x_1<\cdots<x_{n}<\infty} Pf(W(x_1,\cdots,x_{n}))\prod_{1\leq i<j\leq n}(x_j-x_i)\prod_{j=1}^{n} x_j^ke^{tx_j}\rho(x_i)dx_1\cdots dx_{n},
\end{align*}
where the skew-symmetric matrix $W$ is defined as that in Theorem \ref{th:exis}.
\end{remark}

 \subsubsection{A full-discrete  gLV lattice in $3$ dimension}\label{subsec:disLV}
In this subsection, we will present a discrete counterpart of the content in the previous subsection. To this end, we shall consider the deformation on the weight $\rho(x;k,l)dx=x^k(1+x)^l\rho(x)dx$  leading to the $k,l$-index moments:
\begin{align*}
&\mu_{i,j}^{k,l}=\iint_{\mathbb{R}_+^2}x^{k+i}y^{k+j}\omega(x,y)(1+x)^l(1+y)^l\rho(x)\rho(y)dxdy,\\
&\beta_i^{k,l}=\int_{\mathbb{R_+}}x^{k+i}(1+x)^l\rho(x)dx.
\end{align*}
And, define a sequence of monic PSOPs ($\{P_n^{k,l}(z)\}_{n=0}^\infty, k,l\in \mathbb{N}$):
\begin{subequations}\label{psop_dis_lv}
\begin{align} 
&P_{2n}^{k,l}(z)=\frac{1}{\tau_{2n}^{k,l}}Pf(0,1,\cdots,2n-1,2n,z)^{k,l}, \label{psop_dis_lv_even}\\
&P_{2n+1}^{k,l}(z)=\frac{1}{\tau_{2n+1}^{k,l}}Pf(d_0,0,1,\cdots,2n,2n+1,z)^{k,l},\label{psop_dis_lv_odd}
\end{align}
\end{subequations}
where
\begin{align}
\tau_{2n}^{k,l}\triangleq Pf(0,1,\cdots,2n-1)^{k,l}\neq 0,\qquad \tau_{2n+1}^{k,l} \triangleq Pf(d_0,0,1,\cdots,2n)^{k,l}\neq0, \label{psop_dis_lv_tau}
\end{align}
with the Pfaffian entries
\begin{align*}
&Pf(i,j)^{k,l}=\mu_{i,j}^{k,l},\qquad Pf(d_0,i)^{k,l}=\beta_i^{k,l},\\
&Pf(i,z)^{k,l}=z^i,\qquad\ \  Pf(d_0,z)^{k,l}=0.
\end{align*}

Observe that there exist some interrelation among $\tau_n^{k,l}$.
 \begin{lemma}\label{lem:dis_lv1}
 There hold the following relations for the $\tau_n^{k,l}$ defined in \eqref{psop_dis_lv_tau}:
  \begin{align}
 & \tau_{2n}^{k+1,l}=Pf(1,2,\cdots,2n)^{k,l},\qquad\qquad\qquad\ \tau_{2n+1}^{k+1,l}= Pf(d_0,1,2,\cdots,2n+1)^{k,l}, \label{psop_dis_lv_tau_rela1}\\
  &\tau_{2n}^{k,l+1}=Pf(c_0,0,1,\cdots,2n)^{k,l},\qquad\qquad\ \ \ \tau_{2n+1}^{k,l+1}= Pf(d_0,c_0,0,1,\cdots,2n+1)^{k,l},\label{psop_dis_lv_tau_rela2}\\
  &\tau_{2n}^{k+1,l+1}=-Pf(c_0,1,\cdots,2n+1)^{k,l},\qquad\ \  \tau_{2n+1}^{k+1,l+1}=-Pf(d_0,c_0,1,\cdots,2n+2)^{k,l},\label{psop_dis_lv_tau_rela3}
  \end{align}
and also, 
  \begin{align}
&z{\tau_{2n}^{k+1,l}}P_{2n}^{k+1,l}=Pf(1,\cdots,2n+1,z)^{k,l},\label{psop_dis_lv_rel11}\\
&z{\tau_{2n+1}^{k+1,l}}P_{2n+1}^{k+1,l}=Pf(d_0,1,\cdots,2n+2,z)^{k,l},\label{psop_dis_lv_rel1}\\
&(z+1)\tau_{2n}^{k,l+1}P_{2n}^{k,l+1}=Pf(c_0,0,\cdots,2n+1,z)^{k,l},\label{psop_dis_lv_tau_evo_even}\\
&(z+1)\tau_{2n+1}^{k,l+1}P_{2n+1}^{k,l+1}=Pf(d_0,c_0,0,\cdots,2n+2,z)^{k,l}.\label{psop_dis_lv_tau_evo_odd}
\end{align}
Here 
\begin{align*}
Pf(c_0,i)^{k,l}=(-1)^i, \quad Pf(c_0,d_0)^{k,l}=0, \quad Pf(c_0,z)^{k,l}=0.
\end{align*}
 \end{lemma}
 \begin{proof} 
The conclusion \eqref{psop_dis_lv_tau_rela1} , \eqref{psop_dis_lv_rel11} and \eqref{psop_dis_lv_rel1} may be easily obtained by using the definition of the Pfaffian and 
noticing that 
 $$\mu_{i,j}^{k+1,l}=\mu_{i+1,j+1}^{k,l},\qquad \beta_i^{k+1,l}=\beta_{i+1}^{k,l}.$$
 
Furthermore, it is not hard to see that 
 $$
\mu_{i,j}^{k,l+1}=\mu_{i,j}^{k,l}+\mu_{i+1,j}^{k,l}+\mu_{i,j+1}^{k,l}+\mu_{i+1,j+1}^{k,l},\qquad \beta_i^{k,l+1}=\beta_{i+1}^{k,l}+\beta_{i}^{k,l},
 $$
from which we can derive \eqref{psop_dis_lv_tau_rela2} by induction. The relation \eqref{psop_dis_lv_tau_rela3} can be established based on \eqref{psop_dis_lv_tau_rela1} and  \eqref{psop_dis_lv_tau_rela2} and their proofs.

Now we proceed the proofs to \eqref{psop_dis_lv_tau_evo_even} and \eqref{psop_dis_lv_tau_evo_odd}. First, we have 
\begin{align*}
&(z+1)Pf(0,\cdots,2n,z)^{k,l+1}=(z+1)\sum_{i=0}^{2n}(-z)^iPf(0,\cdots,\hat{i},\cdots,2n)^{k,l+1}\\
&=(z+1)\sum_{i=0}^{2n}(-z)^i\sum_{j\neq i}\sum_{\epsilon_j=0,1}Pf(0+\epsilon_0,\cdots,\widehat{i+\epsilon_i},\cdots,2n+\epsilon_{2n})^{k,l}\\
&=\sum_{j=0}^{2n}\sum_{\epsilon_j=0,1}\sum_{i=0}^{2n}(-1)^iz^{i+\epsilon_i}Pf(0+\epsilon_0,\cdots,\widehat{i+\epsilon_i},\cdots,2n+\epsilon_{2n})^{k,l}\\
&=\sum_{j=0}^{2n}\sum_{\epsilon_j=0,1}Pf(0+\epsilon_0,\cdots,2n+\epsilon_{2n},z)^{k,l}
\\&=Pf(c_0,0,\cdots,2n+1,z)^{k,l},
\end{align*}
where the last step is a consequence of induction. Thus the validity of \eqref{psop_dis_lv_tau_evo_even} is  confirmed. By using the similar technique, one can also verify \eqref{psop_dis_lv_tau_evo_odd}. 
 \end{proof}

The relation between adjacent families of the polynomials $\{P_{n}^{k+1,l}\}_{n=0}^\infty$ and $\{P_{n}^{k,l}\}_{n=0}^\infty$ are given in the following corollary.
\begin{coro}
There holds the following relation between adjacent families of the polynomials $\{P_{n}^{k+1,l}\}_{n=0}^\infty$ and $\{P_{n}^{k,l}\}_{n=0}^\infty$:
\begin{align}
(z+1)\tau_{n+1}^{k+1,l}\tau_{n}^{k,l+1}P_{n}^{k,l+1}&+\tau_{n}^{k+1,l+1}\tau_{n+1}^{k,l}P_{n+1}^{k,l}=\nonumber\\
&z(z+1)\tau_{n+2}^{k,l}\tau_{n-1}^{k+1,l+1}P_{n-1}^{k+1,l+1}+z\tau_{n+1}^{k,l+1}\tau_{n}^{k+1,l}P_{n}^{k+1,l}. \label{psop_dis_lv_rel_adj}
\end{align}
\end{coro}
\begin{proof} By employ  Lemma \ref{lem:dis_lv1}, we can rewrite \eqref{psop_dis_lv_rel_adj} in terms of uniform index on $k,l$ as
\begin{align*}
&Pf(d_0,c_0,0,1,\cdots,2n,z)^{k,l}Pf(1,\cdots,2n)^{k,l}=Pf(d_0,c_0,1,\cdots,2n)^{k,l}Pf(0,1,\cdots,2n,z)^{k,l}\\
&\ -Pf(c_0,1,\cdots,2n,z)^{k,l}Pf(d_0,0,1,\cdots,2n)^{k,l}+Pf(d_0,1,\cdots,2n,z)^{k,l}Pf(c_0,0,1,\cdots,2n)^{k,l},\\
&Pf(d_0,c_0,1,\cdots,2n-1,z)^{k,l}Pf(0,1,\cdots,2n-1)^{k,l}=Pf(d_0,0,1,\cdots,2n-1,z)^{k,l}Pf(c_0,1,\cdots,2n-1)^{k,l}\\
&\ -Pf(c_0,0,1,\cdots,2n-1,z)^{k,l}Pf(d_0,1,\cdots,2n-1)^{k,l}+Pf(d_0,c_0,0,1,\cdots,2n-1)^{k,l}Pf(1,\cdots,2n-1,z)^{k,l},
\end{align*}
in the even and odd cases, respectively. These two equalities are valid since they are no other than Pfaffian identities in \eqref{pf1} and \eqref{pf2}.
\end{proof}

Comparing the coefficients of the highest degree on the both sides of \eqref{psop_dis_lv_rel_adj} , we obtain
\begin{align}\label{eq:21glv_dis_bi}
\tau_{n+2}^{k,l}\tau_{n-1}^{k+1,l+1}-\tau_{n}^{k,l+1}\tau_{n+1}^{k+1,l}=\tau_{n}^{k+1,l+1}\tau_{n+1}^{k,l}-\tau_{n}^{k+1,l}\tau_{n+1}^{k,l+1},
\end{align}
which belongs to the discrete BKP hierarchy \cite{gilson2003two,hirota2001soliton}.
Note that this bilinear equation is a discrete analogue of \eqref{eq:bi_glv}, since \eqref{eq:bi_glv} can be reproduced from \eqref{eq:21glv_dis_bi} by taking appropriate limits with respect to the index $l$.  In fact, by setting
$$\hat \tau_n^k(t)=\epsilon^{n+k}\tau_n^{k,l},\qquad t=l\epsilon,$$
one is led to 
$$\hat\tau_{n+2}^{k}(t)\hat\tau_{n-1}^{k+1}(t+\epsilon)-\hat\tau_{n}^{k}(t+\epsilon)\hat\tau_{n+1}^{k+1}(t)=\frac{1}{\epsilon}\left(\hat\tau_{n}^{k+1}(t+\epsilon)\hat\tau_{n+1}^{k}(t)-\hat\tau_{n}^{k+1}(t)\hat\tau_{n+1}^{k}(t+\epsilon)\right),$$
which yields \eqref{eq:bi_glv} in the small limit of $\epsilon$.

%More exactly, it follows from the Taylor expansion $$\tau_{n}^{k,l}=\tau_n^k(t+\epsilon l)$$

If we introduce the variables
\begin{align*}
v_n^{k,l}=\frac{\tau_{n+1}^{k,l}}{\tau_n^{k,l}},\qquad u_n^{k,l}=\frac{\tau_{n}^{k+1,l}}{\tau_n^{k,l}}, \qquad w_n^{k,l}=\frac{\tau_n^{k,l+1}}{\tau_n^{k,l}},
\end{align*}
then $u_n^{k,l}$, $v_n^{k,l}$  and $w_n^{k,l}$ satisfy
\begin{align}
&v_n^{k+1,l}u_n^{k,l}=v_n^{k,l}u_{n+1}^{k,l},\quad v_n^{k,l+1}w_n^{k,l}=v_n^{k,l}w_{n+1}^{k,l},\nonumber\\
&\frac{v_{n+1}^{k,l}}{v_{n-1}^{k+1,l+1}}=1+\frac{u_{n+1}^{k,l}}{u_n^{k,l+1}}-\frac{w_{n+1}^{k,l}}{w_n^{k+1,l}}, \label{eq:21glv_dis}\\
&\ n,l,k=0,1,2,\cdots, \nonumber
\end{align}
which we call ``full-discrete generalized Lotka-Voterra lattice in 3 dimension'' ( full-discrete gLV lattice). This makes sense because its bilinear form is a discrete analogue of the bilinear form of the semi-discrete 1+2 gLV lattice \eqref{eq:21glv}.

The above derivation implies the following theorem, which concludes the tau-function representation of the full-discrete gLV lattice  \eqref{eq:21glv_dis}.
\begin{theorem}
The equation  \eqref{eq:21glv_dis} admits the following tau-function representation
\begin{align*}
v_n^{k,l}=\frac{\tau_{n+1}^{k,l}}{\tau_n^{k,l}},\qquad u_n^{k,l}=\frac{\tau_{n}^{k+1,l}}{\tau_n^{k,l}}, \qquad w_n^{k,l}=\frac{\tau_n^{k,l+1}}{\tau_n^{k,l}},
\end{align*}
which enjoy the explicit forms in terms of Pfaffians:
$$
\tau_{2n}^{k,l}=Pf(0,1,\cdots,2n-1)^{k,l},\qquad \tau_{2n+1}^{k,l} = Pf(d_0,0,1,\cdots,2n)^{k,l},
$$
with the Pfaffian entries satisfying
\begin{align*}
&Pf(i,j)^{k,l}=\mu_{i,j}^{k,l},&&Pf(d_0,i)^{k,l}=\beta_i^{k,l},\\
&\mu_{i,j}^{k+1,l}=\mu_{i+1,j+1}^{k,l},&& \beta_i^{k+1,l}=\beta_{i+1}^{k,l},\\
&\mu_{i,j}^{k,l+1}=\mu_{i,j}^{k,l}+\mu_{i+1,j}^{k,l}+\mu_{i,j+1}^{k,l}+\mu_{i+1,j+1}^{k,l},&&\beta_i^{k,l+1}=\beta_{i+1}^{k,l}+\beta_{i}^{k,l}.
\end{align*}
\end{theorem}
\begin{remark}
Again, the de Bruijn's formulae \eqref{id_deBr1}-\eqref{id_deBr2} imply that the tau-function here admits an equivalent integral formula:
\begin{align*}
\tau_{n}^{k,l}
&=\idotsint\limits_{0<x_1<\cdots<x_{n}<\infty} Pf(W(x_1,\cdots,x_{n}))\prod_{1\leq i<j\leq n}(x_j-x_i)\prod_{j=1}^{n} x_j^k(1+x_j)^l\rho(x_i)dx_1\cdots dx_{n},
\end{align*}
where the skew-symmetric matrix $W$ is defined as that in Theorem \ref{th:exis}.
\end{remark}

\subsection{Toda lattices of BKP type}\label{subsec:btoda}
This subsection is devoted to the derivations of Toda lattices of BKP type. More exactly, five semi(or full)-discrete lattices are considered in different dimensions. Note that the skew-symmetric kernel $\omega(x,y)$ is now assigned as some specific functions, in contrast with that in the subsection \ref{subsec:glv}.
 
 \subsubsection{A semi-discrete B-Toda lattice in $2+1$ dimension}\label{subsec:21btoda}
Consider the special skew-symmetric integral  kernel $\omega(x,y)=\frac{q(x)-q(y)}{q(x)+q(y)}$ and the deformation on the weight $\rho(x;t,s)dx=\exp(tx+sq(x))\rho(x)dx$  leading to the moments $\mu_{i,j}(t,s)$ and $\beta_i(t,s)$ depending on $t,s$:
%\begin{subequations}
\begin{align*}
&\mu_{i,j}(t,s)=\iint_{\mathbb{R}_+^2}x^{i}y^{j}\frac{q(x)-q(y)}{q(x)+q(y)}\exp(t(x+y)+s(q(x)+q(y)))\rho(x)\rho(y)dxdy,\\
&\beta_i(t,s)=\int_{\mathbb{R_+}}x^{i}\exp(tx+sq(x))\rho(x)dx,
\end{align*}
%\end{subequations}
where
$q(x)$ is some ``good'' function on $x$ so that the moments are well-defined \footnote{The reason that $q(x)$ is introduced is to make a distinction between the variables $t$ and $s$.  If the $q(x)$ is a positive and strictly monotone increasing function on $\mathbb{R}_+$, then the validity of the corresponding PSOPs is implied by using Theorem \ref{th:exis} and the Schur's Pfaffian identities \eqref{id_schur} and  \eqref{id_schur_odd}. As we will see in the following  subsection, the case $q(x)=x$ corresponds to a reduction in a lower dimension. 
}.  

This setup implies
$$
\frac{\partial \mu_{i,j}}{\partial s}=\frac{\partial \beta_i}{\partial s}\beta_j-\beta_i\frac{\partial \beta_j}{\partial s},
$$
which also leads to
\begin{align*}
\mu_{i,j}&=\iint\limits_{-\infty<\eta<\sigma<s}\frac{\partial \beta_i}{\partial s}(t;\sigma)\frac{\partial \beta_j}{\partial s}(t;\eta)-\frac{\partial \beta_i}{\partial s}(t;\eta)\frac{\partial \beta_j}{\partial s}(t;\sigma)d\eta d\sigma\\
&=\int_{-\infty}^s\int_{-\infty}^s sgn(\sigma-\eta)\frac{\partial \beta_i}{\partial s}(t;\sigma)\frac{\partial \beta_j}{\partial s}(t;\eta) d\sigma d\eta.
\end{align*}
Define a class of monic PSOPs ($\{P_n(z;t,s)\}_{n=0}^\infty$):
\begin{subequations}\label{psop_btoda}
\begin{align} 
&P_{2n}(z;t,s)=\frac{1}{\tau_{2n}}Pf(0,1,\cdots,2n-1,2n,z), \label{psop_btoda_even}\\
&P_{2n+1}(z;t,s)=\frac{1}{\tau_{2n+1}}Pf(d_0,0,1,\cdots,2n,2n+1,z),\label{psop_btoda_odd}
\end{align}
\end{subequations}
where
\begin{align}
\tau_{2n}(t,s)\triangleq Pf(0,1,\cdots,2n-1)\neq 0,\qquad \tau_{2n+1}(t,s) \triangleq Pf(d_0,0,1,\cdots,2n)\neq0, \label{psop_btoda_tau}
\end{align}
with the Pfaffian entries
\begin{align*}
&Pf(i,j)=\mu_{i,j}(t,s),\qquad\quad Pf(d_0,i)=\beta_i(t,s),\\
&Pf(i,z)=z^i,\qquad\qquad\quad\ \ Pf(d_0,z)=0.
\end{align*}

 We claim that the partial derivative of $\tau_n$ with respect to $t$ and $s$ can both be explicitly expressed in terms of Pfaffians.
 \begin{lemma} \label{lem:psop_btoda_tau_evo}
 The $\tau_n$  evolves as follows:
\begin{align}
&\frac{\partial}{\partial{t}}\tau_{2n}=Pf(0,1,\cdots,2n-2,2n),\qquad\quad\ \frac{\partial}{\partial{t}} \tau_{2n+1}= Pf(d_0,0,1,\cdots,2n-1,2n+1),\label{psop_btoda_tau_evo_t}\\ 
&\frac{\partial}{\partial{s}}\tau_{2n}=Pf(d_0,d_1,0,1,\cdots,2n-1),\qquad \frac{\partial}{\partial{s}} \tau_{2n+1}= Pf(d_1,0,1,\cdots,2n),\label{psop_btoda_tau_evo_s}
\end{align}
where 
$$Pf(d_1,i)= \frac{\partial}{\partial{s}}Pf(d_0,i)=\int_{\mathbb{R_+}}x^{i}q(x)\exp(tx+sq(x))\rho(x)dx,\quad Pf(d_0,d_1)=0.$$
 \end{lemma}
 \begin{proof} 
 It is not hard to see that 
\begin{align*}
&\frac{\partial}{\partial{t}}Pf(i,j)=Pf(i+1,j)+Pf(i,j+1),&& \frac{\partial}{\partial{t}} Pf(d_0,i)=Pf(d_0,i+1),\\
&\frac{\partial}{\partial{s}}Pf(i,j)=Pf(d_0,d_1,i,j),&& \frac{\partial}{\partial{s}} Pf(d_0,i)=Pf(d_1,i),
\end{align*}
from which, the result follows by employing the derivative formula \eqref{der1}, \eqref{der1_odd},  \eqref{der2_1} and  \eqref{der2_2}. 
 \end{proof}
  
 From this lemma, we can express the coefficient of the second highest degree of PSOPs \eqref{psop_btoda} in terms of $\tau_n$ and obtain some nice representations for the derivative with respect to $s$.
\begin{coro}\label{coro:btoda1}
The PSOPs in \eqref{psop_btoda} admit the following expressions:
\begin{align}
P_n(z;t,s)=z^n-(\frac{\partial}{\partial t}\log\tau_n )z^{n-1}+\text{lower order terms}, \label{psop_btoda_exp_tau}
\end{align}
and evolve with respect to $s$ as:
\begin{align}
&\frac{\partial}{\partial{s}}(\tau_{2n}P_{2n})=Pf(d_0,d_1,0,1,\cdots,2n,z),\quad \frac{\partial}{\partial{s}} (\tau_{2n+1}P_{2n+1})= Pf(d_1,0,1,\cdots,2n+1,z),\label{psop_btoda_evo_s}
\end{align}
where 
$$
Pf(d_0,z)=Pf(d_1,z)=0.
$$
\end{coro}
 \begin{proof}
The \eqref{psop_btoda_exp_tau} is an obvious consequence of expanding the expression \eqref{psop_btoda} and using \eqref{psop_btoda_tau_evo_t}.
\eqref{psop_btoda_evo_s} can be obtained by using the expansion formulae and derivative formuae \eqref{der2_1} and \eqref{der2_2}.

 \end{proof}
 
 Moreover, we have
 \begin{coro}\label{coro:btoda2}
 There holds the following relationship between $P_{n}(z;t,s)$ and $P_{n-1}(z;t,s)$:
 \begin{align}
\tau_{n}^2\frac{\partial}{\partial s}P_{n}=\left(\tau_{n-1}\frac{\partial}{\partial s}\tau_{n+1} -\tau_{n+1}\frac{\partial}{\partial s}\tau_{n-1}\right)P_{n-1}-\tau_{n+1}\tau_{n-1}\frac{\partial}{\partial s}P_{n-1}. \label{psop_btoda_rel_adj}
\end{align}
 \end{coro}
\begin{proof} 
By using the relation \eqref{psop_btoda_evo_s} in Corollary \ref{coro:btoda1}, it is not difficult to see that \eqref{psop_btoda_rel_adj} can be equivalently written as
\begin{align*}
&Pf(d_0,d_1,0,\cdots,2n-1,2n,z)Pf(0,\cdots,2n-1)=Pf(d_0,d_1,0,\cdots,2n-1)Pf(0,\cdots,2n-1,2n,z)\\
&\ -Pf(d_0,0,\cdots,2n-1,2n)Pf(d_1,0,\cdots,2n-1,z)+Pf(d_0,0,\cdots,2n-1,z)Pf(d_1,0,\cdots,2n-1,2n),\\
&Pf(d_0,d_1,0,\cdots,2n,z)Pf(0,\cdots,2n,2n+1)=Pf(d_0,d_1,0,\cdots,2n,2n+1)Pf(0,\cdots,2n,z)\\
&\ -Pf(d_1,0,\cdots,2n,2n+1,z)Pf(d_0,0,\cdots,2n)+Pf(d_0,0,\cdots,2n,2n+1,z)Pf(d_1,0,\cdots,2n)
\end{align*}
in the even and odd cases, respectively. They are nothing but the Pfaffian identities in \eqref{pf1} and \eqref{pf2}, thus we complete the proof.
\end{proof}
 
 Inserting \eqref{psop_btoda_exp_tau} into  \eqref{psop_btoda_rel_adj} and comparing the coefficients of the highest degree in $z$ on the two sides, we get a bilinear differential-difference equation in 2+1 dimension (two continuous and one discrete variables), that is,
 \begin{align}
\left(\frac{\partial^2}{\partial t\partial s}\tau_n\right) \tau_n-\frac{\partial}{\partial t}\tau_n \frac{\partial}{\partial s}\tau_n=\left(\frac{\partial}{\partial s}\tau_{n-1}\right)\tau_{n+1}-\left(\frac{\partial}{\partial s}\tau_{n+1}\right)\tau_{n-1}. \label{eq:bi21btoda}
 \end{align}
% \textcolor{red}{*************************************************************************************}
 \begin{remark}
The bilinear equation \eqref{eq:bi21btoda} was first proposed as a special case of the so-called modified Toda lattice of BKP type in \cite{hirota2001soliton}, which enjoys multi-soliton solutions in terms of Pfaffians. Its B\"acklund transformation and Lax pair were studied in \cite{gilson2003two}. Note that we focus on the so-called ``molecule solution'' rather than soliton solution throughout the paper.
% \textcolor{red}{Lax pair??? nonlinear form or bilinear}
 \end{remark}

If we introduce the variables
\begin{align*}
u_n=\frac{\tau_{n+1}\tau_{n-1}}{\tau_n^2},\qquad b_n=\frac{\partial}{\partial s}\log\tau_n,
\end{align*}
then $u_n$ and $b_n$ satisfy
\begin{align}
&\frac{\partial}{\partial s} u_n=u_n(b_{n+1}-2b_n+b_{n-1}),\qquad \frac{\partial}{\partial t} b_n=u_n(b_{n-1}-b_{n+1}), \quad  n=1,2,\cdots,\label{eq:21btoda}
\end{align}
which we call `` Toda lattice of BKP type in 2+1 dimension '' (2+1 B-Toda lattice).

The above derivation implies the following theorem, which concludes the tau-function representation of the 2+1 B-Toda lattice  \eqref{eq:21btoda}.
\begin{theorem}
The 2+1 B-Toda lattice  \eqref{eq:21btoda} admits the following tau-function representation
\begin{align*}
u_n=\frac{\tau_{n+1}\tau_{n-1}}{\tau_n^2},\qquad b_n=\frac{\partial}{\partial s}\log\tau_n,
\end{align*}
which enjoys the explicit form in terms of Pfaffians:
$$
\tau_{2n}(t,s)=Pf(0,1,\cdots,2n-1),\qquad \tau_{2n+1}(t,s) = Pf(d_0,0,1,\cdots,2n),
$$
with the Pfaffian entries satisfying
\begin{align*}
&Pf(i,j)=\mu_{i,j},\qquad\qquad\quad Pf(d_0,i)=\beta_i,&& Pf(d_0,d_1)=0,\\
&\frac{\partial}{\partial t}Pf(d_0,i)=Pf(d_0,i+1),&&
\frac{\partial}{\partial s}Pf(d_0,i)=Pf(d_1,i),\\
&\frac{\partial}{\partial t}Pf(i,j)=Pf(i+1,j)+Pf(i,j+1),&&
\frac{\partial}{\partial s}Pf(i,j)=Pf(d_0,d_1,i,j).
\end{align*}
\end{theorem}
\begin{remark} 
Employing the de Brujin's formulae \eqref{id_deBr1}-\eqref{id_deBr2} and the Schur's Pfaffian identities \eqref{id_schur}-\eqref{id_schur_odd}, it is not difficult to see that the tau-function of 2+1 B-Toda lattice also owns a multiple integral formula
\begin{align*}
\tau_n(t,s)=\idotsint\limits_{0<x_1<\cdots<x_n<\infty}\prod_{1\leq i<j\leq n}\frac{(q(x_i)-q(x_j))(x_j-x_i)}{(q(x_i)+q(x_j))}{\prod_{1\leq j\leq n}}e^{tx_j+sq(x_j)}\rho(x_j)d x_1\cdots dx_n.
\end{align*}
\end{remark}
 
 \subsubsection{A semi-discrete B-Toda lattice in 1+1 dimension} \label{subsec:11btoda}
As a reduction in 1+1 dimension for the subsection \ref{subsec:21btoda},  in the case of $q(x)=x$, we consider the moments $\mu_{i,j}(t)$ \footnote{Note that the  skew-symmetric inner kernel $\frac{x-y}{x+y}$ has appeared in the study of random matrices with Bures ensemble \cite{forrester2016relating}. Again, we remind the reader that the validity of the corresponding PSOPs is guaranteed by using Theorem \ref{th:exis} and the Schur's Pfaffian identities \eqref{id_schur} and  \eqref{id_schur_odd}. }
 depending on $t$:
\begin{align*}
&\mu_{i,j}(t)=\iint_{\mathbb{R}_+^2}x^{i}y^{j}\frac{x-y}{x+y}\exp(t(x+y))\rho(x)\rho(y)dxdy,\\
&\beta_i(t)=\int_{\mathbb{R_+}}x^{i}\exp(tx)\rho(x)dx.
\end{align*}
This setup implies
$$
\frac{d \mu_{i,j}}{d t}=\frac{d \beta_i}{d t}\beta_j-\beta_i\frac{d \beta_j}{d t},
$$
which also leads to
\begin{align*}
\mu_{i,j}
=\int_{-\infty}^s\int_{-\infty}^s sgn(\sigma-\eta)\frac{d \beta_i}{d t}(t=\sigma)\frac{d \beta_j}{d t}(t=\eta) d\sigma d\eta.
\end{align*}
%&=\iint\limits_{-\infty<\eta<\sigma<s}\frac{d \beta_i}{d t}(t=\sigma)\frac{d \beta_j}{d t}(t=\eta)-\frac{d \beta_i}{d t}(t=\eta)\frac{d \beta_j}{d t}(t=\sigma)d\eta d\sigma\\
Let's consider a class of monic PSOPs ($\{P_n(z;t)\}_{n=0}^\infty$):
\begin{subequations}\label{psop_11btoda}
\begin{align} 
&P_{2n}(z;t)=\frac{1}{\tau_{2n}}Pf(0,1,\cdots,2n-1,2n,z), \label{psop_11btoda_even}\\
&P_{2n+1}(z;t)=\frac{1}{\tau_{2n+1}}Pf(d_0,0,1,\cdots,2n,2n+1,z),\label{psop_11btoda_odd}
\end{align}
\end{subequations}
where
\begin{align}
\tau_{2n}(t)\triangleq Pf(0,1,\cdots,2n-1)\neq 0,\qquad \tau_{2n+1}(t) \triangleq Pf(d_0,0,1,\cdots,2n)\neq0, \label{psop_11btoda_tau}
\end{align}
with the Pfaffian entries
\begin{align*}
&Pf(i,j)=\mu_{i,j}(t),\qquad\ \ Pf(d_0,i)=\beta_i(t),\\
&Pf(i,z)=z^i,\qquad\qquad\  Pf(d_0,z)=0.
\end{align*}
It is noted  that 
\begin{align*}
&\frac{d}{d{t}}Pf(i,j)=Pf(i+1,j)+Pf(i,j+1)=Pf(d_0,d_1,i,j),\\
&\frac{d}{d{t}}Pf(d_0,i)=Pf(d_0,i+1)=Pf(d_1,i),
\end{align*}
where 
$$Pf(d_1,i)=\beta_{i+1},\quad Pf(d_0,d_1)=0.$$
Similar to Lemma \ref{lem:psop_btoda_tau_evo},  we immediately obtain that there exist two kinds of representations for the derivative of $\tau_n$ with respect to $t$ by applying the derivative formula \eqref{der1}, \eqref{der1_odd},  \eqref{der2_1} and  \eqref{der2_2}. 
\begin{lemma} \label{lem:psop_11btoda_tau_evo}
 The $\tau_n$  evolves as follows:
 \begin{subequations}
\begin{align}
&\frac{d}{d{t}}\tau_{2n}=Pf(0,1,\cdots,2n-2,2n)=Pf(d_0,d_1,0,1,\cdots,2n-1), \label{psop_11btoda_tau_evo_even}\\
& \frac{d}{d{t}} \tau_{2n+1}= Pf(d_0,0,1,\cdots,2n-1,2n+1)=Pf(d_1,0,1,\cdots,2n).\label{psop_11btoda_tau_evo_odd}
\end{align}
\end{subequations}
 \end{lemma}

By using this lemma, we can express the coefficient of the second highest degree of PSOPs \eqref{psop_11btoda} in terms of $\tau_n$ and obtain some nice representations for the derivative with respect to $t$. 
\begin{coro}
The PSOPs in \eqref{psop_11btoda} admit the following expressions:
\begin{align}
P_n(z;t)=z^n-(\frac{d}{d t}\log\tau_n )z^{n-1}+\text{lower order terms}, \label{psop_11btoda_exp_tau}
\end{align}
and evolve with respect to $t$ as:
\begin{align}
&\frac{\partial}{\partial{t}}(\tau_{2n}P_{2n})=Pf(d_0,d_1,0,1,\cdots,2n,z),\quad \frac{\partial}{\partial{t}} (\tau_{2n+1}P_{2n+1})= Pf(d_1,0,1,\cdots,2n+1,z),\label{psop_11btoda_evo_t1}
\end{align}
where 
$$
Pf(d_0,z)=Pf(d_1,z)=0.
$$
\end{coro}
 \begin{proof}
 The proof can be completed by following that for Corollary \eqref{coro:btoda1}.
 \end{proof}
 
 Moreover, by following the proof for  Corollary \eqref{coro:btoda2}, we easily get
 \begin{coro}\label{coro:btoda_evo4}
 There holds the following relationship between $P_{n}(z;t)$ and $P_{n+1}(z;t)$:
 \begin{align}
\tau_{n}^2\frac{\partial}{\partial t}P_{n}=\left( \tau_{n-1}\frac{d}{d t}\tau_{n+1}-\tau_{n+1}\frac{d}{d t}\tau_{n-1}\right)P_{n-1}-\tau_{n+1}\tau_{n-1}\frac{\partial}{\partial t}P_{n-1}. \label{psop_11btoda_rel_adj}
\end{align}
 \end{coro}

Inserting \eqref{psop_11btoda_exp_tau} into  \eqref{psop_11btoda_rel_adj} and comparing the coefficients of the highest degree in $z$ on the two sides, we get a bilinear ODE system in 1+1 dimension, that is,
 \begin{align}
\left(\frac{d^2}{d t^2}\tau_n\right) \tau_n-\left( \frac{d}{d t}\tau_n\right)^2=\left(\frac{d}{d t}\tau_{n-1}\right)\tau_{n+1}-\left(\frac{d}{d t}\tau_{n+1}\right)\tau_{n-1}. \label{eq:bi11btoda}
 \end{align}
% \textcolor{red}{*************************************************************************************}
 
If we introduce the variables
\begin{align*}
u_n=\frac{\tau_{n+1}\tau_{n-1}}{\tau_n^2},\qquad b_n=\frac{d}{d t}\log\tau_n,
\end{align*}
then $u_n$ and $b_n$ satisfy
\begin{align}
&\frac{d}{d t} u_n=u_n(b_{n+1}-2b_n+b_{n-1}),\qquad \frac{d}{d t} b_n=u_n(b_{n-1}-b_{n+1}), \quad  n=1,2,\cdots,\label{eq:11btoda}
\end{align}
which we call `` Toda lattice of BKP type in 1+1 dimension '' (1+1 B-Toda lattice).

In summary, we can conclude that the following theorem.
\begin{theorem}
The 1+1 B-Toda lattice  \eqref{eq:11btoda} admits the following tau-function representation
\begin{align*}
u_n=\frac{\tau_{n+1}\tau_{n-1}}{\tau_n^2},\qquad b_n=\frac{d}{d t}\log\tau_n,
\end{align*}
which enjoys the explicit form in terms of Pfaffians:
$$
\tau_{2n}(t)=Pf(0,1,\cdots,2n-1),\qquad \tau_{2n+1}(t) = Pf(d_0,0,1,\cdots,2n),
$$
with the Pfaffian entries satisfying
\begin{align*}
&Pf(i,j)=\mu_{i,j},\qquad Pf(d_0,i)=\beta_i ,\qquad Pf(d_0,d_1)=0,\\
&\frac{d}{d t}Pf(d_0,i)=Pf(d_0,i+1)=Pf(d_1,i),\\
&\frac{d}{d t}Pf(i,j)=Pf(i+1,j)+Pf(i,j+1)=Pf(d_0,d_1,i,j).
\end{align*}
\end{theorem}

\begin{remark} 
Applying the de Brujin's formulae \eqref{id_deBr1}-\eqref{id_deBr2} and the Schur's Pfaffian identities \eqref{id_schur}-\eqref{id_schur_odd}, it is not difficult to see that the tau-function of 1+1 B-Toda lattice owns a matrix integral solution
\begin{align*}
\tau_n(t)=(-1)^{\lfloor\frac{n}{2}\rfloor}\idotsint\limits_{0<x_1<\cdots<x_n<\infty}\prod_{1\leq i<j\leq n}\frac{(x_j-x_i)^2}{(x_i+x_j)}{\prod_{1\leq j\leq n}}e^{tx_j}\rho(x_j)d x_1\cdots dx_n,
\end{align*}
where $\lfloor a\rfloor$ denotes the greatest integer less than or equal to $a$.
As is shown in \cite{forrester2016relating}, this matrix integral appears as the partition function for the Bures ensemble. Therefore, the tau-function of B-Toda lattice could be viewed as a t-deformation of the partition function for Bures ensemble, just as stated in \cite{hu2017partition}.
\end{remark}
\begin{remark}
In \cite{chang2017application}, a finite truncation of this 1+1 dimensional B-Toda lattice has been introduced and was shown to enjoy an intriguing connection with  multipeakons of the Novikov equation. 
\end{remark}

We end this section by providing a Lax pair of 1+1 B-Toda lattice \eqref{eq:11btoda}, before which we demonstrate there exists a four-term recurrence for the PSOPs \eqref{psop_11btoda} yielding the Lax matrix. 
First, we have
 \begin{lemma} \label{lem:psop_11btoda_evo_t2}
 There holds 
 \begin{subequations}\label{psop_11btoda_evo_t2}
 \begin{align}
& (z+\frac{d}{dt})(\tau_{2n}P_{2n})=Pf(0,1,\cdots,2n-1,2n+1,z),\\
 & (z+\frac{d}{dt})(\tau_{2n+1}P_{2n+1})=Pf(d_0,0,1,\cdots,2n,2n+2,z).
 \end{align}
 \end{subequations}
 \end{lemma}
 \begin{proof}
 Let's give a detailed proof for the second equality, i.e. the odd case.  The even case is similar. \footnote{Actually, a proof can be found in \cite[Lemma 3.6]{adler1999pfaff}.}
 \begin{align*}
 & (z+\frac{d}{dt})(\tau_{2n+1}P_{2n+1})\\
=& \sum_{j=0}^{2n+1}(-1)^{j+1}z^{j+1}Pf(d_0,0,1,\cdots,\hat j,\cdots, 2n+1)+\sum_{j=0}^{2n+1}(-1)^{j+1}z^{j+1}\frac{d}{dt}Pf(d_0,0,1,\cdots,\hat j,\cdots, 2n+1)\\
= & \sum_{j=0}^{2n+1}(-1)^{j+1}z^{j+1}Pf(d_0,0,1,\cdots,\hat j,\cdots, 2n+1)+Pf(d_0,1,2,\cdots,2n,2n+2)+\\
& \sum_{j=1}^{2n}(-1)^{j+1}z^{j}\left(Pf(d_0,0,1,\cdots,\hat j-1,\cdots, 2n,2n+1)+Pf(d_0,0,1,\cdots,\hat j,\cdots, 2n,2n+2)\right)+\\
&z^{2n+1}Pf(d_0,0,1,\cdots,2n-1,2n+1)\\
=&z^{2n+2}Pf(d_0,0,1,\cdots,2n)+\sum_{j=0}^{2n}(-1)^{j+1}z^jPf(d_0,0,1,\cdots,\hat j,\cdots 2n,2n+2)\\
=&Pf(d_0,0,1,\cdots,2n,2n+2,z),
 \end{align*}
 where the derivative formula \eqref{der1_gen} is employed in the second equality, thus the proof is completed.
 
 \end{proof}
 %\textcolor{red}{Is there any similar equalities  for 2+1 B-Toda???}
 
 From this lemma and by employing the Pfaffian identities, we obtain
 \begin{coro} \label{coro:psop_11btoda_evo_t3}
 There holds
 \begin{align}\label{psop_11btoda_evo_t3}
 (z+\frac{d}{d t})P_{n}=P_{n+1}+\left(\frac{d}{d t}\log\frac{\tau_{n+1}}{\tau_{n}}\right)P_{n}-\frac{\tau_{n-1}\tau_{n+2}}{\tau_{n}\tau_{n+1}}P_{n-1}.
 \end{align}
 \end{coro}
 \begin{proof}
 The required Pfaffian identities are
\begin{align*}
&Pf(d_0,0,\cdots,2n+1,z)Pf(0,\cdots,2n-1)=Pf(d_0,0,\cdots,2n)Pf(0,\cdots,2n-1,2n+1,z)\\
&\ -Pf(d_0,0,\cdots,2n-1,2n+1)Pf(0,\cdots,2n,z)+Pf(d_0,0,\cdots,2n-1,z)Pf(0,\cdots,2n+1),\\
&Pf(d_0,0,\cdots,2n+2)Pf(0,\cdots,2n,z)=Pf(0,\cdots,2n+2,z)Pf(d_0,0,\cdots,2n)\\
&\ -Pf(0,\cdots,2n,2n+1)Pf(d_0,0,\cdots,2n,2n+2,z)+Pf(0,\cdots,2n,2n+2)Pf(d_0,0,\cdots,2n+1,z),
\end{align*}
by using which and Lemma \ref{lem:psop_11btoda_evo_t2}, one can get the conclusion.
 \end{proof}

Combining the results in Corollary \ref{coro:btoda_evo4} and \ref{coro:psop_11btoda_evo_t3}, we finally get
\begin{theorem}\label{psop_4term}
The PSOPs \eqref{psop_11btoda} satisfy a four-term recurrence relationship
\begin{align}\label{fourterm1}
z(P_n+u_nP_{n-1})=P_{n+1}+(b_{n+1}-b_{n}+u_n)P_n-u_n(b_{n+1}-b_{n}+u_{n+1})P_{n-1}-u_n^2u_{n-1}P_{n-2},
\end{align}
where $u_n=\frac{\tau_{n+1}\tau_{n-1}}{\tau_{n}^2}$ and $b_n=\frac{d}{d t}\log \tau_n$.
\end{theorem}
\begin{proof}
This is a consequence of a straightforward computation of eliminating the terms including the derivative term $\frac{d}{dt}P_n$ and $\frac{d}{dt}P_{n-1}$ in \eqref{psop_11btoda_rel_adj}   by using \eqref{psop_11btoda_evo_t3}.
\end{proof}

Now we are ready to present the Lax formula. Rewrite \eqref{fourterm1} and \eqref{psop_11btoda_rel_adj} as:
$$z \Phi=L\Phi,\qquad \Phi_t=B\Phi,$$
where 
\begin{align*}
&\qquad\qquad\qquad\Phi=(P_0,P_1,\cdots)^\top,\qquad L=L_1^{-1}L_2,\qquad B=L_1^{-1}B_2,\\
&L_1=\left(\begin{array}{cccccc}
1&&\\
u_1&1&\\
&u_2&1\\
%&&\ddots&1\\
&&\ddots&\ddots
\end{array}
\right),\quad
B_2=\left(\begin{array}{cccccc}
0&&&&\\
u_1b_2&0&&&\\
&u_2(b_3-b_1)&0\\
&&\ddots&\ddots
\end{array}
\right),\\
&
L_2=\left(\begin{array}{cccccc}
b_1&1&&&\\
u_1(b_1-b_2-u_2)&b_2-b_1+u_1&1&&\\
-u_2^2u_1&u_2(b_2-b_3-u_3)&b_3-b_2+u_2&1\\
&\ddots&\ddots&\ddots&\ddots
\end{array}
\right).
\end{align*}
The compatibility condition $\dot L=[B,L]$ gives the 1+1 B-toda lattice \eqref{eq:11btoda} with the boundary condition $u_0=b_0=0$.

 \subsubsection{A semi-discrete B-Toda lattice in 1+2 dimension} \label{subsec:12btoda}
 In this subsection, we plan to discretize the continuous variable ``$s$'' for the 2+1 B-Toda lattice to obtain a integrable lattice in 1+2 dimension. To this end,
 we will deal with the special skew-symmetric integral  kernel $\omega(x,y)=\frac{x-y}{xy+x+y}$ and the deformation on the weight $\rho(x;t,k)dx=(1+x)^k\exp(tx)\rho(x)dx$  leading to the moments $\mu_{i,j}^k(t)$ and $\beta_i^k(t)$ depending on $t$:
\begin{align*}
&\mu_{i,j}^k(t)=\iint_{\mathbb{R}_+^2}x^{i}y^{j}\frac{x-y}{xy+x+y}(1+x)^k(1+y)^k\exp(t(x+y))\rho(x)\rho(y)dxdy,\\
&\beta_i^k(t)=\int_{\mathbb{R_+}}x^{i}(1+x)^k\exp(tx)\rho(x)dx.
\end{align*}

Define a class of monic PSOPs ($\{P_n(z;t)\}_{n=0}^\infty$):
\begin{subequations}\label{psop_12btoda}
\begin{align} 
&P_{2n}^k(z;t)=\frac{1}{\tau_{2n}}Pf(0,1,\cdots,2n-1,2n,z)^k, \label{psop_12btoda_even}\\
&P_{2n+1}^k(z;t)=\frac{1}{\tau_{2n+1}}Pf(d_0,0,1,\cdots,2n,2n+1,z)^k,\label{psop_12btoda_odd}
\end{align}
\end{subequations}
where
\begin{align}
\tau_{2n}^k(t)\triangleq Pf(0,1,\cdots,2n-1)^k\neq 0,\qquad \tau_{2n+1}^k(t) \triangleq Pf(d_0,0,1,\cdots,2n)^k\neq0, \label{psop_12btoda_tau}
\end{align}
with the Pfaffian entries
\begin{align*}
&Pf(i,j)^k=\mu_{i,j}^k(t),\qquad \ \ Pf(d_0,i)=\beta_i^k(t),\\
&Pf(i,z)^k=z^i,\qquad\qquad\   Pf(d_0,z)^k=0.
\end{align*}
 
 Then it is not difficult to obtain the following interrelation for $\tau_n^k(t)$ and $P_n^k(z;t)$ by using the argument before. 
  \begin{lemma} \label{lem:psop_12btoda_tau_evo}
 The $\tau_n^k$  evolves as follows:
\begin{align}
&\frac{d}{d{t}}\tau_{2n}^k=Pf(0,1,\cdots,2n-2,2n)^k,&& \frac{d}{d{t}} \tau_{2n+1}= Pf(d_0,0,1,\cdots,2n-1,2n+1)^k,\label{psop_12btoda_tau_evo_t}\\ 
&\tau_{2n}^{k+1}=\tau_{2n}^k+Pf(d_0,d_1,0,1,\cdots,2n-1)^k,&&  \tau_{2n+1}^{k+1}=\tau_{2n+1}^k+ Pf(d_1,0,1,\cdots,2n)^k,\label{psop_12btoda_tau_evo_k}
\end{align}
and there hold
\begin{subequations}
\begin{align}
&\frac{d}{d{t}}\tau_{2n}^k=Pf(0,1,\cdots,2n-2,2n)^k,&& \frac{d}{d{t}} \tau_{2n+1}= Pf(d_0,0,1,\cdots,2n-1,2n+1)^k,\label{psop_12btoda_k1}\\ 
&\tau_{2n}^{k+1}=\tau_{2n}^k+Pf(d_0,d_1,0,1,\cdots,2n-1)^k,&&  \tau_{2n+1}^{k+1}=\tau_{2n+1}^k+ Pf(d_1,0,1,\cdots,2n)^k.\label{psop_12btoda_k2}
\end{align}
\end{subequations}
Furthermore, we have 
 \begin{subequations}\label{psop_12btoda_evo_k}
\begin{align}
&\tau_{2n}^{k+1}P_{2n}^{k+1}=\tau_{2n}^{k}P_{2n}^{k}+Pf(d_0,d_1,0,1,\cdots,2n,z)^k,\\
& \tau_{2n+1}^{k+1}P_{2n+1}^{k+1}=\tau_{2n+1}^{k}P_{2n+1}^{k}+ Pf(d_1,0,1,\cdots,2n+1,z)^k.
\end{align}
\end{subequations}
Here 
$$Pf(d_1,i)^k=Pf(d_0,i+1)^k,\qquad Pf(d_0,d_1)^k=0,\qquad Pf(d_0,z)=Pf(d_1,z)=0.$$
 \end{lemma}
 \begin{proof} 
\eqref{psop_12btoda_tau_evo_t} follows from
\begin{align*}
&\frac{d}{d{t}}Pf(i,j)^k=Pf(i+1,j)^k+Pf(i,j+1)^k,&& \frac{d}{d{t}} Pf(d_0,i)^k=Pf(d_0,i+1)^k
\end{align*}
and the derivative formulae \eqref{der1}, \eqref{der1_odd}. Noticing that 
$$Pf(i,j)^{k+1}=Pf(i,j)^k+(d_0,d_1,i,j)^k,\qquad Pf(d_0,i)^{k+1}=Pf(d_0,i)^k+Pf(d_1,i)^k,$$
one obtains \eqref{psop_12btoda_tau_evo_k} by use of the addition formula \eqref{add_g1}-\eqref{add_g2}. At last, \eqref{psop_12btoda_evo_k} can be obtained by using the expansion formulae and the addition formula \eqref{add_g1}-\eqref{add_g2}.
 \end{proof}
 
 By expanding the expression \eqref{psop_12btoda} and using \eqref{psop_12btoda_tau_evo_t}, we immediately get
\begin{coro}\label{coro:12btoda1}
The PSOPs in \eqref{psop_12btoda} admit the following expressions:
\begin{align}
P_n^k(z;t)=z^n-(\frac{d}{d t}\log\tau_n^k )z^{n-1}+\text{lower order terms}. \label{psop_12btoda_exp_tau}
\end{align}
\end{coro}
 
 Moreover, we have
 \begin{coro}\label{coro:12btoda2}
 There holds the following relationship between $P_{n}(z;t)$ and $P_{n-1}(z;t)$:
 \begin{align}
\tau_{n}^{k+1}\tau_n^k(P_n^{k+1}-P_n^k)=\tau_{n+1}^{k+1}\tau_{n-1}^{k}P_{n-1}^{k}-\tau_{n+1}^k\tau_{n-1}^{k+1}P_{n-1}^{k+1}. \label{psop_12btoda_rel_adj}
\end{align}
 \end{coro}
\begin{proof} 
The key is to notice the following Pfaffian identities
\begin{align*}
&Pf(d_0,d_1,0,\cdots,2n-1,2n,z)^kPf(0,\cdots,2n-1)^k=Pf(d_0,d_1,0,\cdots,2n-1)^kPf(0,\cdots,2n-1,2n,z)^k\\
&\ -Pf(d_0,0,\cdots,2n-1,2n)^kPf(d_1,0,\cdots,2n-1,z)^k+Pf(d_0,0,\cdots,2n-1,z)^kPf(d_1,0,\cdots,2n-1,2n)^k,\\
&Pf(d_0,d_1,0,\cdots,2n,z)^kPf(0,\cdots,2n,2n+1)^k=Pf(d_0,d_1,0,\cdots,2n,2n+1)^kPf(0,\cdots,2n,z)^k\\
&\ -Pf(d_1,0,\cdots,2n,2n+1,z)^kPf(d_0,0,\cdots,2n)^k+Pf(d_0,0,\cdots,2n,2n+1,z)^kPf(d_1,0,\cdots,2n)^k.
\end{align*}
Then the conclusion follows by using the relations in Lemma  \ref{lem:psop_12btoda_tau_evo}.
\end{proof}
 
 Inserting \eqref{psop_12btoda_exp_tau} into  \eqref{psop_12btoda_rel_adj} and comparing the coefficients of the highest degree in $z$ on the two sides, we get a bilinear differential-difference equation in 1+2 dimension (one continuous and two discrete variables), that is,
 \begin{align}
\left(\frac{d}{dt}\tau_n^{k+1}\right)\tau_n^k-\left(\frac{d}{dt}\tau_n^{k}\right)\tau_n^{k+1}=\tau_{n+1}^k\tau_{n-1}^{k+1}-\tau_{n+1}^{k+1}\tau_{n-1}^{k}. \label{eq:bi12btoda}
 \end{align}
 \begin{remark} Let
 $$\hat \tau_n(t,s)=\delta^{-\frac{n^2}{2}}\tau_n^{k}(t),\qquad s=k\delta.$$
 In the small limit of $\delta$, one can arrive at the bilinear 2+1 B-Toda lattice \eqref{eq:bi21btoda} from \eqref{eq:bi12btoda}.
 \end{remark}
 
If we introduce the variables
\begin{align*}
r_n^k=\frac{\tau_n^{k+1}}{\tau_n^k},\qquad v_n^k=\frac{\tau_{n+1}^k}{\tau_{n}^{k+1}}, 
\end{align*}
then $r_n^k$ and $v_n^k$ satisfy
\begin{align}
\frac{d}{dt}r_n^k=&\frac{v_n^kr_{n}^k}{v_{n-1}^kr_{n-1}^k}(r_{n-1}^k-r_{n+1}^k),\qquad v_n^kr_{n}^k=r_{n+1}^{k-1}v_n^{k-1}, \quad n,k=0,1,2,\cdots,\label{eq:12btoda}
\end{align}
which we call `` Toda lattice of BKP type in 1+2 dimension '' (1+2 B-Toda lattice).

The above derivation implies the tau-function representation of the 1+2 B-Toda lattice  \eqref{eq:12btoda}.
\begin{theorem}
The 1+2 B-Toda lattice  \eqref{eq:12btoda} admits the following tau-function representation
\begin{align*}
r_n^k=\frac{\tau_n^{k+1}}{\tau_n^k},\qquad v_n^k=\frac{\tau_{n+1}^k}{\tau_{n}^{k+1}}, 
\end{align*}
which enjoys the explicit form in terms of Pfaffians:
$$
\tau_{2n}^k(t)=Pf(0,1,\cdots,2n-1)^k,\qquad \tau_{2n+1}^k(t) = Pf(d_0,0,1,\cdots,2n)^k,
$$
with the Pfaffian entries satisfying
\begin{align*}
&Pf(i,j)^k=\mu_{i,j}^k,\qquad\qquad\quad  Pf(d_0,d_1)^k=0,&& Pf(d_1,i-1)^k=Pf(d_0,i)^k=\beta_i^k,\\
&\frac{d}{d{t}}Pf(i,j)^k=Pf(i+1,j)^k+Pf(i,j+1)^k,&& \frac{d}{d{t}} Pf(d_0,i)^k=Pf(d_0,i+1)^k,\\
&Pf(i,j)^{k+1}=Pf(i,j)^k+(d_0,d_1,i,j)^k,&&Pf(d_0,i)^{k+1}=Pf(d_0,i)^k+Pf(d_1,i)^k.
\end{align*}
\end{theorem}
\begin{remark} 
Again, employing the de Brujin's formulae \eqref{id_deBr1}-\eqref{id_deBr2} and the Schur's Pfaffian identities \eqref{id_schur}-\eqref{id_schur_odd}, one can get a multiple integral formula of the tau-function of 1+2 B-Toda lattice
\begin{align*}
\tau_n^k(t)=(-1)^{\lfloor\frac{n}{2}\rfloor}\idotsint\limits_{0<x_1<\cdots<x_n<\infty}\prod_{1\leq i<j\leq n}\frac{(x_j-x_i)^2}{x_ix_j+x_i+x_j}{\prod_{1\leq j\leq n}}(1+x_j)^ke^{tx_j}\rho(x_j)d x_1\cdots dx_n.
\end{align*}
\end{remark}

\subsubsection{A full-discrete B-Toda lattice in 3 dimension}\label{subsec:3btoda}
  In this subsection, we will present a full-discrete B-Toda lattice. To this end, let's consider the skew-symmetric kernel $\omega(x,y)=\frac{x-y}{xy+x+y}$ and the deformation on the weight $\rho(x;k,l)dx=(1+x)^{k+l}\rho(x)dx$  leading to the $k,l$-index moments:
\begin{align*}
&\mu_{i,j}^{k,l}=\iint_{\mathbb{R}_+^2}x^{i}y^{j}\frac{x-y}{xy+x+y}(1+x)^{k+l}(1+y)^{k+l}\rho(x)\rho(y)dxdy,\\
&\beta_i^{k,l}=\int_{\mathbb{R_+}}x^{i}(1+x)^{k+l}\rho(x)dx.
\end{align*}
And, define a sequence of monic PSOPs ($\{P_n^{k,l}(z)\}_{n=0}^\infty, k,l\in \mathbb{N}$):
\begin{subequations}\label{psop_dis_3btoda}
\begin{align} 
&P_{2n}^{k,l}(z)=\frac{1}{\tau_{2n}^{k,l}}Pf(0,1,\cdots,2n-1,2n,z)^{k,l}, \label{psop_dis_3btoda_even}\\
&P_{2n+1}^{k,l}(z)=\frac{1}{\tau_{2n+1}^{k,l}}Pf(d_0,0,1,\cdots,2n,2n+1,z)^{k,l},\label{psop_dis_3btoda_odd}
\end{align}
\end{subequations}
where
\begin{align}
\tau_{2n}^{k,l}\triangleq Pf(0,1,\cdots,2n-1)^{k,l}\neq 0,\qquad \tau_{2n+1}^{k,l} \triangleq Pf(d_0,0,1,\cdots,2n)^{k,l}\neq0, \label{psop_dis_3btoda_tau}
\end{align}
with the Pfaffian entries
\begin{align*}
&Pf(i,j)^{k,l}=\mu_{i,j}^{k,l},\qquad \ \ \ \ \  Pf(d_0,i)^{k,l}=\beta_i^{k,l},\\
&Pf(i,z)^{k,l}=z^i,\qquad\qquad\   Pf(d_0,z)^{k,l}=0.
\end{align*}
 \begin{lemma}\label{lem:dis_3btoda1}
 There hold the following relations for the $\tau_n^{k,l}$ defined in \eqref{psop_dis_3btoda_tau}:
  \begin{eqnarray*}
 && \tau_{2n}^{k+1,l}=\tau_{2n}^{k,l}+Pf(d_0,d_1,0,\cdots,2n-1)^{k,l},\quad \tau_{2n+1}^{k+1,l}= \tau_{2n+1}^{k,l}+ Pf(d_1,0,\cdots,2n)^{k,l}, \label{psop_dis_3btoda_tau_rela1}\\
  &&\tau_{2n}^{k,l+1}=Pf(c_0,0,1,\cdots,2n)^{k,l},\qquad\qquad\quad\quad \tau_{2n+1}^{k,l+1}= Pf(d_0,c_0,0,1,\cdots,2n+1)^{k,l},\label{psop_dis_3btoda_tau_rela2}\\
  &&\tau_{2n}^{k+1,l+1}=\tau_{2n}^{k,l+1}+Pf(d_0,d_1,c_0,0,\cdots,2n)^{k,l},\label{psop_dis_3btoda_tau_rela3_even}\\
  && \tau_{2n+1}^{k+1,l+1}=\tau_{2n+1}^{k,l+1}+Pf(d_1,c_0,0,\cdots,2n+1)^{k,l},\label{psop_dis_3btoda_tau_rela3}
  \end{eqnarray*}
and also, 
  \begin{eqnarray*}
&&\tau_{2n}^{k+1,l}P_{2n}^{k+1,l}=\tau_{2n}^{k,l}P_{2n}^{k,l}+Pf(d_0,d_1,0,1,\cdots,2n,z)^{k,l},\label{psop_dis_3btoda_rel1_even}\\
&& \tau_{2n+1}^{k+1,l}P_{2n+1}^{k+1,l}=\tau_{2n+1}^{k,l}P_{2n+1}^{k,l}+ Pf(d_1,0,1,\cdots,2n+1,z)^{k,l},\label{psop_dis_3btoda_rel1_odd}\\
&&(z+1)\tau_{2n}^{k,l+1}P_{2n}^{k,l+1}=Pf(c_0,0,\cdots,2n+1,z)^{k,l},\label{psop_dis_3btoda_tau_evo_even}\\
&&(z+1)\tau_{2n+1}^{k,l+1}P_{2n+1}^{k,l+1}=Pf(d_0,c_0,0,\cdots,2n+2,z)^{k,l},\label{psop_dis_3btoda_tau_evo_odd}\\
&&(z+1)\tau_{2n}^{k+1,l+1}P_{2n}^{k,l+1}=(z+1)\tau_{2n}^{k,l+1}P_{2n}^{k,l+1}+Pf(d_0,d_1,c_0,0,\cdots,2n+1,z)^{k,l},\label{psop_dis_3btoda_tau_evo_even_2}\\
&&(z+1)\tau_{2n+1}^{k,l+1}P_{2n+1}^{k+1,l+1}=(z+1)\tau_{2n+1}^{k,l+1}P_{2n+1}^{k,l+1}+Pf(d_1,c_0,0,\cdots,2n+2,z)^{k,l}.\label{psop_dis_3btoda_tau_evo_odd_2}
  \end{eqnarray*}
Here 
\begin{align*}
&Pf(d_1,i)^{k,l}=Pf(d_0,i+1),\quad Pf(c_0,i)^{k,l}=(-1)^i, \quad Pf(c_0,z)^{k,l}=0,\\
& Pf(c_0,d_0)^{k,l}=Pf(c_0,d_1)^{k,l}= Pf(d_0,d_1)^{k,l}=0.
\end{align*}
 \end{lemma}
 \begin{proof} 
 All of these formulae can be proved by using the argument as before. The key is to observe that
\begin{align*}
&Pf(i,j)^{k+1,l}=Pf(i,j)^{k,l}+(d_0,d_1,i,j)^{k,l},\qquad Pf(d_0,i)^{k+1,l}=Pf(d_0,i)^{k,l}+Pf(d_1,i)^{k,l},\\
&Pf(i,j)^{k,l+1}=Pf(i,j)^{k,l}+(i+1,j)^{k,l}+(i,j+1)^{k,l}+(i+1,j+1)^{k,l},\\
&Pf(d_0,i)^{k,l+1}=Pf(d_0,i)^{k,l}+Pf(d_0,i+1)^{k,l}.
\end{align*}
And then one can get the conclusions by induction or employing the formulae \eqref{add_g1}-\eqref{add_w2}. 
 \end{proof}
Subsequently, we have 
\begin{coro}
There holds the following relation among adjacent families of the polynomials $\{P_{n}^{k,l}\}_{n=0}^\infty$:
\begin{align}
\tau_{n}^{k+1,l+1}\tau_{n}^{k,l}P_{n}^{k,l}-\tau_{n}^{k,l+1}\tau_{n}^{k+1,l}P_{n}^{k+1,l}=(z+1)\left(\tau_{n+1}^{k,l}\tau_{n-1}^{k+1,l+1}P_{n-1}^{k+1,l+1}-\tau_{n+1}^{k+1,l}\tau_{n-1}^{k,l+1}P_{n-1}^{k,l+1}\right). \label{psop_dis_3btoda_rel_adj}
\end{align}
\end{coro}
\begin{proof} 
The key is to observe the Pfaffian identities
\begin{align*}
&Pf(d_0,d_1,c_0,0,\cdots,2n)^{k,l}Pf(0,\cdots,2n,z)^{k,l}=Pf(d_0,d_1,0,\cdots,2n,z)^{k,l}Pf(c_0,0,\cdots,2n)^{k,l}\\
&\ -Pf(d_0,c_0,0,\cdots,2n,z)^{k,l}Pf(d_1,0,\cdots,2n)^{k,l}+Pf(d_1,c_0,0,\cdots,2n,z)^{k,l}Pf(d_0,0,\cdots,2n)^{k,l},\\
&Pf(d_0,d_1,c_0,0,\cdots,2n+1,z)^{k,l}Pf(0,\cdots,2n+1)^{k,l}=Pf(d_0,d_1,0,\cdots,2n+1)^{k,l}Pf(c_0,0,\cdots,2n+1,z)^{k,l}\\
&\ -Pf(d_0,c_0,0,\cdots,2n+1)^{k,l}Pf(d_1,0,\cdots,2n+1,z)^{k,l}+Pf(d_0,0,\cdots,2n+1,z)^{k,l}Pf(d_1,c_0,0,\cdots,2n+1)^{k,l},
\end{align*}
by using which and Lemma \ref{lem:dis_3btoda1}, one can achieve the goal.
\end{proof}

Comparing the coefficients of the highest degree on the both sides of \eqref{psop_dis_3btoda_rel_adj} , we obtain
\begin{align}\label{eq:21g3btoda_dis_bi}
\tau_{n}^{k+1,l+1}\tau_{n}^{k,l}-\tau_{n}^{k,l+1}\tau_{n}^{k+1,l}=\tau_{n-1}^{k+1,l+1}\tau_{n+1}^{k,l}-\tau_{n+1}^{k+1,l}\tau_{n-1}^{k,l+1}.
\end{align}
 \begin{remark} Let
 $$\hat \tau_n^k(t)=\epsilon^{-\frac{n^2}{2}}\tau_n^{k,l},\qquad t=l\epsilon.$$
 In the small limit of $\epsilon$, one can arrive at the bilinear form \eqref{eq:bi12btoda} of 1+2 B-Toda lattice from \eqref{eq:21g3btoda_dis_bi}.
 \end{remark}

If we introduce the variables
\begin{align*}
v_n^{k,l}=\frac{\tau_{n+1}^{k,l}}{\tau_n^{k,l}},\qquad u_n^{k,l}=\frac{\tau_{n}^{k+1,l}}{\tau_n^{k,l}}, \qquad w_n^{k,l}=\frac{\tau_n^{k,l+1}}{\tau_n^{k,l}},
\end{align*}
then $u_n^{k,l}$, $v_n^{k,l}$  and $w_n^{k,l}$ satisfy
\begin{align}
&v_n^{k+1,l}u_n^{k,l}=v_n^{k,l}u_{n+1}^{k,l},\quad v_n^{k,l+1}w_n^{k,l}=v_n^{k,l}w_{n+1}^{k,l},\nonumber\\
&\frac{u_{n+1}^{k,l+1}}{u_{n+1}^{k,l}}=1+\frac{v_{n+1}^{k,l}w_{n}^{k+1,l}}{w_{n+1}^{k,l}v_{n}^{k+1,l}}-\frac{v_{n+1}^{k+1,l}}{w_{n}^{k,l+1}}, \label{eq:21g3btoda_dis}\\
&\ n,l,k=0,1,2,\cdots, \nonumber
\end{align}
which we call ``full-discrete Toda lattice of BKP type in three dimension'' ( full-discrete 3D B-Toda lattice). 

The above derivation implies the following theorem, which concludes the tau-function representation of the full-discrete 3D B-Toda lattice  \eqref{eq:21g3btoda_dis}.
\begin{theorem}
The equation  \eqref{eq:21g3btoda_dis} admits the following tau-function representation
\begin{align*}
v_n^{k,l}=\frac{\tau_{n+1}^{k,l}}{\tau_n^{k,l}},\qquad u_n^{k,l}=\frac{\tau_{n}^{k+1,l}}{\tau_n^{k,l}}, \qquad w_n^{k,l}=\frac{\tau_n^{k,l+1}}{\tau_n^{k,l}},
\end{align*}
which enjoy the explicit forms in terms of Pfaffians:
$$
\tau_{2n}^{k,l}=Pf(0,1,\cdots,2n-1)^{k,l},\qquad \tau_{2n+1}^{k,l} = Pf(d_0,0,1,\cdots,2n)^{k,l},
$$
with the Pfaffian entries satisfying
\begin{align*}
&Pf(i,j)^{k,l}=\mu_{i,j}^{k,l},&&Pf(d_0,i)^{k,l}=\beta_i^{k,l},\\
&\mu_{i,j}^{k+1,l}=\mu_{i,j}^{k,l}+Pf(d_0,d_1,i,j)^{k,l},&& \beta_i^{k+1,l}=\beta_{i+1}^{k,l}+\beta_{i}^{k,l},\\
&\mu_{i,j}^{k,l+1}=\mu_{i,j}^{k,l}+\mu_{i+1,j}^{k,l}+\mu_{i,j+1}^{k,l}+\mu_{i+1,j+1}^{k,l},&&\beta_i^{k,l+1}=\beta_{i+1}^{k,l}+\beta_{i}^{k,l}.
\end{align*}
\end{theorem}
\begin{remark}
Again, by using the de Bruijn's formulae \eqref{id_deBr1}-\eqref{id_deBr2} and  the Schur's Pfaffian identities \eqref{id_schur}-\eqref{id_schur_odd}, one is led to a multiple integral formula of the tau-function of full-discrete  3D B-Toda lattice
\begin{align*}
\tau_n^{k,l}=(-1)^{\lfloor\frac{n}{2}\rfloor}\idotsint\limits_{0<x_1<\cdots<x_n<\infty}\prod_{1\leq i<j\leq n}\frac{(x_j-x_i)^2}{x_ix_j+x_i+x_j}{\prod_{1\leq j\leq n}}(1+x_j)^{k+l}\rho(x_j)d x_1\cdots dx_n.
\end{align*}
\end{remark}
 \begin{remark}
The bilinear 3D B-Toda lattice \eqref{eq:21g3btoda_dis_bi} has appeared in the literature \cite{gilson2003two,hirota2001soliton}. However, this kind of Pfaffian tau-function is never addressed.
 \end{remark} 
 
 \subsubsection{A full-discrete B-Toda lattice in 2 dimension} 
  As a reduction in lower dimension for the subsection \ref{subsec:3btoda} and also a discrete counterpart for the subsection \ref{subsec:11btoda}, we plan to consider the skew-symmetric kernel $\omega(x,y)=\frac{x-y}{xy+x+y}$ and the deformation on the weight $\rho(x;k)dx=(1+x)^{k}\rho(x)dx$  leading to the $k$-index moments:
\begin{align*}
&\mu_{i,j}^{k}=\iint_{\mathbb{R}_+^2}x^{i}y^{j}\frac{x-y}{xy+x+y}(1+x)^{k}(1+y)^{k}\rho(x)\rho(y)dxdy,\\
&\beta_i^{k}=\int_{\mathbb{R_+}}x^{i}(1+x)^{k}\rho(x)dx.
\end{align*}
And, define a sequence of monic PSOPs ($\{P_n^{k}(z)\}_{n=0}^\infty, k\in \mathbb{N}$):
\begin{subequations}\label{psop_dis_2btoda}
\begin{align} 
&P_{2n}^{k}(z)=\frac{1}{\tau_{2n}^{k}}Pf(0,1,\cdots,2n-1,2n,z)^{k}, \label{psop_dis_2btoda_even}\\
&P_{2n+1}^{k}(z)=\frac{1}{\tau_{2n+1}^{k}}Pf(d_0,0,1,\cdots,2n,2n+1,z)^{k},\label{psop_dis_2btoda_odd}
\end{align}
\end{subequations}
where
\begin{align}
\tau_{2n}^{k}\triangleq Pf(0,1,\cdots,2n-1)^{k}\neq 0,\qquad \tau_{2n+1}^{k} \triangleq Pf(d_0,0,1,\cdots,2n)^{k}\neq0, \label{psop_dis_2btoda_tau}
\end{align}
with the Pfaffian entries
\begin{align*}
&Pf(i,j)^{k}=\mu_{i,j}^{k},\qquad\quad\ \   Pf(d_0,i)^{k}=\beta_i^{k},\\
&Pf(i,z)^{k}=z^i,\qquad\qquad\  Pf(d_0,z)^{k}=0.
\end{align*}
 
It is not hard to see that all the derivations in the subsection \ref{subsec:3btoda}  can apply here by a reduction $$\tau_n^{k,l}=\tau_n^m,\ \  P_n^{k,l}=P_n^m, \ \ m=k+l.$$ Consequently, we arrive at
\begin{coro}
There holds the following relation among adjacent families of the polynomials $\{P_{n}^{k}\}_{n=0}^\infty$:
\begin{align}
\tau_{n}^{k+2}\tau_{n}^{k}P_{n}^{k}-(\tau_{n}^{k+1})^2P_{n}^{k+l}=(z+1)\left(\tau_{n+1}^{k}\tau_{n-1}^{k+2}P_{n-1}^{k+2}-\tau_{n+1}^{k+1}\tau_{n-1}^{k+1}P_{n-1}^{k+1}\right). \label{psop_dis_2btoda_rel_adj}
\end{align}
\end{coro}
Comparing the coefficients of the highest degree on the both sides of \eqref{psop_dis_2btoda_rel_adj} , we obtain
\begin{align}\label{eq:2btoda_dis_bi}
\tau_{n}^{k+2}\tau_{n}^{k}-(\tau_{n}^{k+1})^2=\tau_{n+1}^{k}\tau_{n-1}^{k+2}-\tau_{n+1}^{k+1}\tau_{n-1}^{k+1}.
\end{align}
 \begin{remark} Let
 $$\hat \tau_n(t)=\epsilon^{-\frac{n^2}{2}}\tau_n^{k},\qquad t=k\epsilon.$$
 In the small limit of $\epsilon$, one can arrive at the 1+1 B-Toda lattice \eqref{eq:bi11btoda} from \eqref{eq:2btoda_dis_bi}.
 \end{remark}

If we introduce the variables
\begin{align*}
v_n^{k}=\frac{\tau_{n+1}^{k}}{\tau_n^{k}},\qquad u_n^{k}=\frac{\tau_{n}^{k+1}}{\tau_n^{k}},
\end{align*}
then $u_n^k$, $v_n^k$ satisfy
\begin{align}
&v_n^{k+1}u_n^{k}=v_n^{k}u_{n+1}^{k},\quad \frac{u_{n+1}^{k+1}}{u_{n+1}^{k}}=1+\frac{v_{n+1}^{k}u_{n}^{k+1}}{u_{n+1}^{k}v_{n}^{k+1}}-\frac{v_{n+1}^{k+1}}{u_{n}^{k+1}}, \quad n,k=0,1,2,\cdots,\label{eq:2btoda_dis}
\end{align}
which we call ``full-discrete Toda lattice of BKP type in two dimension'' ( full-discrete 2D B-Toda lattice). 

At the end, we present the tau-function representation of the full-discrete 2D B-Toda lattice  \eqref{eq:2btoda_dis}.
\begin{theorem}
The equation  \eqref{eq:2btoda_dis} admits the following tau-function representation
\begin{align*}
v_n^{k}=\frac{\tau_{n+1}^{k}}{\tau_n^{k}},\qquad u_n^{k}=\frac{\tau_{n}^{k+1}}{\tau_n^{k}}, 
\end{align*}
which enjoy the explicit forms in terms of Pfaffians:
$$
\tau_{2n}^{k}=Pf(0,1,\cdots,2n-1)^{k},\qquad \tau_{2n+1}^{k} = Pf(d_0,0,1,\cdots,2n)^{k},
$$
with the Pfaffian entries satisfying
\begin{align*}
&Pf(i,j)^{k}=\mu_{i,j}^{k},&&Pf(d_0,i)^{k}=\beta_i^{k},\\
&\mu_{i,j}^{k+1}=\mu_{i,j}^{k}+Pf(d_0,d_1,i,j)^{k}=\mu_{i,j}^{k}+\mu_{i+1,j}^{k}+\mu_{i,j+1}^{k}+\mu_{i+1,j+1}^{k},&& \beta_i^{k+1}=\beta_{i+1}^{k}+\beta_{i}^{k}.
\end{align*}
\end{theorem}
\begin{remark}
Obviously, as a reduction for that in subsection \ref{subsec:3btoda}, one is led to a multiple integral formula of the tau-function of  the full-discrete  2D B-Toda lattice
\begin{align*}
\tau_n^{k}=(-1)^{\lfloor\frac{n}{2}\rfloor}\idotsint\limits_{0<x_1<\cdots<x_n<\infty}\prod_{1\leq i<j\leq n}\frac{(x_j-x_i)^2}{x_ix_j+x_i+x_j}{\prod_{1\leq j\leq n}}(1+x_j)^{l}\rho(x_j)d x_1\cdots dx_n.
\end{align*}
\end{remark}

 \section{Two integrable lattices as algorithms}\label{sec:alg}
 In the literature, there exist many integrable lattices, which can be used to design numerical algorithms. A recent work by the authors \cite{chang2017new} addresses the acceleration of sequence transformations represented in terms of Pfaffians by using a discrete integrable lattice. Again, recall that Hietarinta, Joshi and Nijhoff claimed that the ``connections between vector Pad\'e approximants and integrable lattices remain largely to be explored'' in  \cite[Section4.4]{hietarinta2016discrete}.
 
 In this section, we derive two full-discrete integrable lattices related to PSOPs, one of which is nothing but the one in \cite{chang2017new} and the other one can be used to compute vector Pad\'{e} approximants in \cite{graves1996cayley,graves1997problems}. Our derivation for these two systems still relies on the methodology in the above section. 
 \subsection{For Pad\'{e} approximants}
 Let's consider the $k,l$-index moments:
\begin{align*}
&\mu_{i,j}^{k,l}=\iint_{\mathbb{R}^2}x^{k-i}y^{k-j}(1+x)^l(1+y)^l\omega(x,y)\rho(x)\rho(y)dxdy,\\
&\beta_i^{k,l}=\int_{\mathbb{R}}x^{k-i}(1+x)^l\rho(x)dx.
\end{align*}
And then define a sequence of the reciprocal PSOPs ($\{P_n^{k,l}(z)\}_{n=0}^\infty, k,l\in \mathbb{N}$):
\begin{subequations}\label{psop_pade}
\begin{align} 
&P_{2n}^{k,l}(z)=\frac{z^{2n}}{\tau_{2n}^{k,l}}Pf(0,1,\cdots,2n-1,2n,z)^{k,l}, \label{psop_pade_even}\\
&P_{2n+1}^{k,l}(z)=\frac{z^{2n+1}}{\tau_{2n+1}^{k,l}}Pf(d_0,0,1,\cdots,2n,2n+1,z)^{k,l},\label{psop_pade_odd}
\end{align}
\end{subequations}
where
\begin{align}
\tau_{2n}^{k,l}\triangleq Pf(0,1,\cdots,2n-1)^{k,l}\neq 0,\qquad \tau_{2n+1}^{k,l} \triangleq Pf(d_0,0,1,\cdots,2n)^{k,l}\neq0, \label{psop_pade_tau}
\end{align}
with the Pfaffian entries
\begin{align*}
&Pf(i,j)^{k,l}=\mu_{i,j}^{k,l},\qquad\qquad\  Pf(d_0,i)^{k,l}=\beta_i^{k,l},\\
&Pf(i,z)^{k,l}=z^{-i},\qquad\qquad\   Pf(d_0,z)^{k,l}=0.
\end{align*}
Note that these polynomials are not monic, but with constant terms $1$. 

Observe that the moments here are very similar those in subsection \ref{subsec:disLV} and there hold
\begin{align*}
&\mu_{i,j}^{k+1,l}=\mu_{i-1,j-1}^{k,l},&&\beta_i^{k+1,l}=\beta_{i-1}^{k,l},\\
&\mu_{i,j}^{k,l+1}=\mu_{i,j}^{k,l}+\mu_{i-1,j}^{k,l}+\mu_{i,j-1}^{k,l}+\mu_{i-1,j-1}^{k,l},&&\beta_i^{k,l+1}=\beta_{i-1}^{k,l}+\beta_{i}^{k,l},
\end{align*}
from which, we can see that the coefficient of the highest degree on $z$ of $P_n^{k,l}$ is $\frac{\tau_{n}^{k-1,l}}{\tau_{n}^{k,l}}$ and the constant term is $1$.

Actually, it is not hard to obtain the following conclusion by induction just as Lemma \ref{lem:dis_lv1}. Here the proof is omitted.
 \begin{lemma}\label{lem:pade1}
 There hold the following relations for the $\tau_n^{k,l}$ defined in \eqref{psop_pade_tau}:
  \begin{eqnarray*}
 && \tau_{2n}^{k+1,l}=Pf(-1,0,\cdots,2n-2)^{k,l},\qquad\quad\ \tau_{2n+1}^{k+1,l}= Pf(d_0,-1,0\cdots,2n-1)^{k,l}, \label{psop_pade_tau_rela1}\\
 & &\tau_{2n}^{k,l+1}=Pf(c_0,-1,0,\cdots,2n-1)^{k,l},\qquad \tau_{2n+1}^{k,l+1}= Pf(d_0,c_0,-1,0,\cdots,2n)^{k,l},\label{psop_pade_tau_rela2}
 % &\tau_{2n}^{k+1,l+1}=-Pf(c_0,1,\cdots,2n+1)^{k,l},\qquad \tau_{2n+1}^{k+1,l+1}=-Pf(d_0,c_0,1,\cdots,2n+2)^{k,l},\label{psop_pade_tau_rela3}
  \end{eqnarray*}
and also, 
  \begin{eqnarray*}
&&{\tau_{2n}^{k+1,l}}P_{2n}^{k+1,l}=z^{2n-1}Pf(-1,0,\cdots,2n-1,z)^{k,l},\label{psop_pade_rel_even}\\
&&{\tau_{2n+1}^{k+1,l}}P_{2n+1}^{k+1,l}=z^{2n}Pf(d_0,-1,0,\cdots,2n,z)^{k,l},\label{psop_pade_rel_odd}\\
&&(z+1)\tau_{2n}^{k,l+1}P_{2n}^{k,l+1}=z^{2n}Pf(c_0,-1,0,\cdots,2n,z)^{k,l},\label{psop_pade_tau_evo_even}\\
&&(z+1)\tau_{2n+1}^{k,l+1}P_{2n+1}^{k,l+1}=z^{2n+1}Pf(d_0,c_0,-1,0,\cdots,2n+1,z)^{k,l}.\label{psop_pade_tau_evo_odd}
  \end{eqnarray*}
Here 
\begin{align*}
Pf(c_0,i)^{k,l}=(-1)^{i-1}, \quad Pf(c_0,d_0)^{k,l}=0, \quad Pf(c_0,z)^{k,l}=0.
\end{align*}
 \end{lemma}
 
Moreover, we get
\begin{coro}
There holds the following relation between adjacent families of the polynomials:
\begin{align}
P_{n+1}^{k,l}=(z+1)P_n^{k-1,l+1}-z\frac{\tau_n^{k,l}\tau_{n+1}^{k-1,l+1}}{\tau_n^{k-1,l+1}\tau_{n+1}^{k,l}}P_n^{k,l}+z(z+1)\frac{\tau_{n+2}^{k,l}\tau_{n-1}^{k-1,l+1}}{\tau_n^{k-1,l+1}\tau_{n+1}^{k,l}}P_{n-1}^{k-1,l+1}.\label{psop_pade_rel_adj}
\end{align}
\end{coro}
\begin{proof} By employ  Lemma \ref{lem:pade1}, we can rewrite \eqref{psop_pade_rel_adj} in terms of uniform index on $k,l$ as
\begin{align*}
&Pf(d_0,c_0,-1,\cdots,2n+1,z)^{k,l}Pf(-1,\cdots,2n)^{k,l}=Pf(d_0,c_0,-1,\cdots,2n)^{k,l}Pf(-1,\cdots,2n+1,z)^{k,l}\\
&\ -Pf(c_0,-1,\cdots,2n,z)^{k,l}Pf(d_0,-1,\cdots,2n+1)^{k,l}+Pf(d_0,-1,\cdots,2n,z)^{k,l}Pf(c_0,-1,\cdots,2n+1)^{k,l},\\
&Pf(d_0,c_0,-1,\cdots,2n-1,z)^{k,l}Pf(-1,\cdots,2n-2)^{k,l}=Pf(d_0,-1,\cdots,2n,z)^{k,l}Pf(c_0,-1,\cdots,2n-1)^{k,l}\\
&\ -Pf(c_0,-1,\cdots,2n,z)^{k,l}Pf(d_0,-1,\cdots,2n-1)^{k,l}+Pf(d_0,c_0,-1,\cdots,2n)^{k,l}Pf(-1,\cdots,2n-1,z)^{k,l},
\end{align*}
in the even and odd cases, respectively. These two equalities are valid since they are no other than Pfaffien identities in \eqref{pf1} and \eqref{pf2}.
\end{proof}

Comparing the coefficients of the highest degree on $z$ at the two sides of \eqref{psop_pade_rel_adj}, we obtain
 \begin{align}\label{eq:bi_pade}
 \tau_{n+2}^{k,l}\tau_{n-1}^{k-2,l+1}=\tau_{n+1}^{k,l}\tau_n^{k-2,l+1}-\tau_{n+1}^{k-1,l+1}\tau_n^{k-1,l}+\tau_{n+1}^{k-1,l}\tau_{n}^{k-1,l+1},
\end{align}
which belongs to the discrete BKP hierarchy \cite{gilson2003two,hirota2001soliton}.
%\textcolor{red}{***************************************************}
If we introduce the variables
% $$
%u_n^{k,l}= \frac{\tau_n^{k,l}\tau_{n+1}^{k-1,l+1}}{\tau_n^{k-1,l+1}\tau_{n+1}^{k,l}},\qquad v_n^{k,l}=\frac{\tau_{n+2}^{k,l}\tau_{n-1}^{k-1,l+1}}{\tau_n^{k-1,l+1}\tau_{n+1}^{k,l}}
%$$
\begin{align*}
v_n^{k,l}=\frac{\tau_{n+1}^{k,l}}{\tau_{n}^{k,l}},\quad u_n^{k,l}=\frac{\tau_n^{k+1,l}}{\tau_{n}^{k,l}},\quad w_n^{k,l}=\frac{\tau_n^{k,l+1}}{\tau_{n}^{k,l}},
\end{align*}
then one obtains a nonlinear system with $u_n^{k,l},v_n^{k,l},w_n^{k,l}$ satisfying
\begin{subequations}\label{eq:nonl_pade}
\begin{align}
&\frac{v_{n+1}^{k+1,l}}{v_{n-1}^{k-1,l+1}}=1-\frac{w_{n+1}^{k,l}u_n^{k-1,l+1}}{u_{n+1}^{k,l}w_n^{k,l}}+\frac{u_n^{k-1,l+1}}{u_n^{k,l}},\\
&u_n^{k,l}w_n^{k,l-1}=w_n^{k+1,l-1}u_n^{k,l-1},\\
&u_{n+1}^{k,l}v_n^{k,l}=v_n^{k+1,l}u_n^{k,l}.
\end{align}
\end{subequations}

%\begin{align*}
%v_n^{k,l}=\frac{\tau_{n+1}^{k,l}}{\tau_n^{k,l}},\qquad u_n^{k,l}=\frac{\tau_{n}^{k+1,l}}{\tau_n^{k,l}}, \qquad r_n^{k,l}=\frac{\tau_n^{k,l+1}}{\tau_n^{k,l}},
%\end{align*}
%then $u_n^k$ $v_n^k$,  and $r_n^k$ satisfy
%\begin{align}
%&v_n^{k+1,l}u_n^{k,l}=v_n^{k,l}u_{n+1}^{k,l},\quad v_n^{k,l+1}r_n^{k,l}=v_n^{k,l}r_{n+1}^{k,l},\nonumber\\
%&\frac{v_{n+1}^{k,l}}{v_{n-1}^{k+1,l+1}}=1+\frac{u_{n+1}^{k,l}}{u_n^{k,l+1}}-\frac{r_{n+1}^{k,l}}{r_n^{k+1,l}} \label{eq:pade_dis}\\
%&\qquad\qquad n=0,1,2,\cdots,\quad l,k=0,\pm1,\pm2\cdots. \nonumber
%\end{align}
%\textcolor{red}{***************************************************}

Next we show how the above results are related to the so-called generalized inverse vector-valued Pad\'e approximant (GIPA) \cite{graves1986vector,graves1996cayley,graves1997problems}. 

For a given real vector of formal power series ${\bm{f}}(z)=(f_1(z),f_2(z),\cdots, f_d(z))^\top$, its GIPA of type $[N/2K]$ is the rational function
$${\bf{R}}(z)={\bf{Q}}(z)/\tilde P(z)$$
defined uniquely by the axiomatic requirements that $\bf{Q}(z)$ is a vector polynomial and $\tilde P(z)$ is a scalar polynomial satisfying
\begin{enumerate}
\item deg $\{{\bf{Q}(z)}\}\leq N$, deg $\{\tilde P(z)\}=2K$,
\item $\tilde P(z)\mid\Vert{{\bf{Q}}(z)}\Vert ^2$,
\item $\tilde P(z){\bf{f}}(z)-{\bf{Q}}(z)=O(z^{N+1})$,
\item $\tilde P(0)\not=0$.
\end{enumerate} 
The existence and uniqueness of this problem was addressed in \cite{graves1986vector}. The solution of this problem has been clarified in \cite{graves1996cayley,graves1997problems}. The key for the solution is the derivation of the denominator $\tilde P(z)$. The corresponding numerator polynomial $\bf{Q}(z)$ could be determined by the third condition. In fact, it was shown that the denominator $\tilde P(z)$ could be written in a Pfaffian form 
\begin{eqnarray*}
\tilde P(z)=z^{2K}Pf(0,\cdots,2K-1)Pf(0,\cdots,2K,z),
\end{eqnarray*}
with Pfaffian elements
\begin{align*}
&Pf(i,j)=\iint_{\mathbb{T}^2}x^{2K-i}y^{2K-j}\omega(x,y)\rho(x)\rho(y)dxdy,\qquad Pf (i,z)=z^{-i},
\end{align*}
where 
$$
w(x,y)=-2\frac{{\bm{f}} (x) {\bm{f}} (y)^\top}{x-y}(xy)^{-N+1},\qquad \rho(x)=\frac{1}{2\pi \bm{I}}x^{-N+1},
$$
and 
$\mathbb{T}$ is the appropriate contour enclosing $x = 0$ in the complex plane, $\bm{I}$ is the imaginary unit. 

It is not hard to see that the denominator $\tilde P(z)$ is nothing but $\left(\tau_{2K}^{2K,0}\right)^2 P_{2K}^{2K,0}(z)$ in \eqref{psop_pade} with the above integral kernel $\omega(x,y)$ and the weight function $\rho(x)$ . If we let $\tilde P_{n}^{k,l}(z)=(\tau_{n}^{k,l})^2P_{n}^{k,l}(z)$, it follows from\eqref{psop_pade_rel_adj} that 
\begin{align*}
\tilde P_{n+1}^{k,l}=(z+1)\frac{(\tau_{n+1}^{k,l})^2}{(\tau_n^{k-1,l+1})^2}\tilde P_n^{k-1,l+1}-z\frac{\tau_{n+1}^{k,l}\tau_{n+1}^{k-1,l+1}}{\tau_n^{k-1,l+1}\tau_{n}^{k,l}}\tilde P_n^{k,l}+z(z+1)\frac{\tau_{n+1}^{k,l}\tau_{n+2}^{k,l}}{\tau_n^{k-1,l+1}\tau_{n-1}^{k-1,l+1}}\tilde P_{n-1}^{k-1,l+1}
%\tilde P_{n+1}^{k,l}=(z+1)\frac{v_n^{k,l}u_{n}^{k-1,l}}{w_n^{k-1,l}}\tilde P_n^{k-1,l+1}-zv_n^{k-1,l+1}\tilde P_n^{k,l}+z(z+1)\frac{v_{n+1}^{k,l}v_n^{k,l}u_{n}^{k-1,l}}{w_n^{k-1,l}}\tilde P_{n-1}^{k-1,l+1}
\end{align*}
with $\tau_n^{k,l}$ satisfying \eqref{eq:bi_pade}. If the initial value is imposed as
\begin{align*}
&\tau_{-1}^{k,l}=0, \quad \tau_0^{k,l}=1,\quad  \tau_1^{k,l}=(d_0,0)^{k,l}, \quad \tau_2^{k,l}=Pf(0,1)^{k,l},\\
&\tilde P_{-1}^{k,l}=0, \qquad \qquad\qquad\qquad\qquad \tilde P_0^{k,l}=1.
\end{align*}
 then the denominator $\tilde P(z)=\tilde P_{2K}^{2K,0}(z)$ can be recursively computed after obtaining the $\tau_n^{k,l}$ recursively.\footnote{Details of the implementation and analysis for the algorithm will be addressed elsewhere. } In conclusion, one can recursively calculate the GIPA for a given vector of formal power series with the help of the integrable lattice \eqref{eq:bi_pade}.

 \subsection{For convergence acceleration}
 Consider the $k,l$-index moments:
\begin{align*}
&\mu_{i,j}^{k,l}=\iint_{\mathbb{R}^2}x^{k+i}y^{k+j}\omega(x,y)(1-x)^l(1-y)^l\rho(x)\rho(y)dxdy,\\
&\beta_i^{k,l}=\iint_{\mathbb{R}^2}x^{i}y^{k+i}\omega(x,y)(1-x)^{l-1}(1-y)^l\rho(x)\rho(y)dxdy.
\end{align*}
We note that $\beta_i^{k,l}$ involves double integrals, which is quite different from the previous cases. 

Now let's define a sequence of monic PSOPs ($\{P_n^{k,l}(z)\}_{n=0}^\infty, k,l\in \mathbb{N}$):
\begin{subequations}\label{psop_conv_acc}
\begin{align} 
&P_{2n}^{k,l}(z)=\frac{1}{\tau_{2n}^{k,l}}Pf(0,1,\cdots,2n-1,2n,z)^{k,l}, \label{psop_conv_acc_even}\\
&P_{2n+1}^{k,l}(z)=\frac{1}{\tau_{2n+1}^{k,l}}Pf(d_0,0,1,\cdots,2n,2n+1,z)^{k,l},\label{psop_conv_acc_odd}
\end{align}
\end{subequations}
where
\begin{align}
\tau_{2n}^{k,l}\triangleq Pf(0,1,\cdots,2n-1)^{k,l}\neq 0,\qquad \tau_{2n+1}^{k,l} \triangleq Pf(d_0,0,1,\cdots,2n)^{k,l}\neq0, \label{psop_conv_acc_tau}
\end{align}
with the Pfaffian entries
\begin{align*}
&Pf(i,j)^{k,l}=\mu_{i,j}^{k,l},\qquad\qquad\  Pf(d_0,i)^{k,l}=\beta_i^{k,l},\\
&Pf(i,z)^{k,l}=z^i,\qquad\qquad\ \ \   Pf(d_0,z)^{k,l}=0.
\end{align*}
Besides, we introduce an additional sequence of polynomials $\{Q_n^{k,l}(z)\}_{n=0}^\infty, k,l\in \mathbb{N}$ defined by
\begin{align}
Q_{2n}^{k,l}(z)=\frac{1}{\sigma_{2n}^{k,l}}Pf(c_0,d_0,0,\cdots,2n,z)^{k,z}, \label{psop_conv_acc_Q}
\end{align}
where 
$$\sigma_{2n}^{k,l}=Pf(c_0,d_0,0,\cdots,2n-1)^{k,l}$$
with the Pfaffian entries
$$Pf(c_0,i)^{k,l}=1,\qquad Pf(c_0,d_0)^{k,l}=0,\qquad Pf(c_0,z)^{k,l}=0.$$
Here we have the convention that $\tau_0^{k,l}=1$, $\sigma_0^{k,l}=0$ and $Q_0^{k,l}=1.$

It is clear to see that 
\begin{align*}
&\mu_{i,j}^{k+1,l}=\mu_{i+1,j+1}^{k,l},&&\beta_i^{k+1,l}=\beta_{i+1}^{k,l}+\mu_{0,i+1}^{k,l},\\
&\mu_{i,j}^{k,l+1}=\mu_{i,j}^{k,l}-\mu_{i+1,j}^{k,l}-\mu_{i,j+1}^{k,l}+\mu_{i+1,j+1}^{k,l},&& \beta_i^{k,l+1}=\mu_{0,i+1}^{k,l}-\mu_{0,i}^{k,l},
\end{align*}
which result in the following conclusion by induction as before. We omit the detailed proof here.
 \begin{lemma}\label{lem:conv_acc1}
 There hold the following relations for the $\tau_n^{k,l}$ defined in \eqref{psop_conv_acc_tau}:
  \begin{eqnarray*}
 && \tau_{2n}^{k+1,l}=Pf(1,2,\cdots,2n)^{k,l},\qquad\quad\ \tau_{2n+1}^{k+1,l}= Pf(d_0,1,2,\cdots,2n+1)^{k,l}+\tau_{2n+2}^{k,l}, \label{psop_conv_acc_tau_rela1}\\
  &&\tau_{2n}^{k,l+1}=Pf(c_0,0,1,\cdots,2n)^{k,l},\qquad \tau_{2n+1}^{k,l+1}=\tau_{2n+2}^{k,l},\label{psop_conv_acc_tau_rela2}\\
  &&\tau_{2n}^{k+1,l+1}=Pf(c_0,1,\cdots,2n+1)^{k,l},\qquad \\
    &&\sigma_{2n+1}^{k+1,l}=-Pf(d_0,c_0,1,\cdots,2n+2)^{k,l}+\tau_{2n}^{k,l+1}-\tau_{2n}^{k+1,l},\label{psop_conv_acc_tau_rela3}
  \end{eqnarray*}
and also, 
%  \begin{align}
%&z{\tau_{2n}^{k+1,l}}P_{2n}^{k+1,l}=Pf(1,\cdots,2n+1,z)^{k,l},\\
% &z{\tau_{2n+1}^{k+1,l}}P_{2n+1}^{k+1,l}=Pf(d_0,1,\cdots,2n+2,z)^{k,l}+\tau_{2n+2}^{k,l}P_{2n+2}^{k,l}+\tau_{2n+2}^{k+1,l},\label{psop_conv_acc_rel1}\\
%&(z-1)\tau_{2n}^{k,l+1}P_{2n}^{k,l+1}=Pf(c_0,0,\cdots,2n+1,z)^{k,l},\label{psop_conv_acc_tau_evo_even}\\
%&(z-1)\tau_{2n+1}^{k,l+1}P_{2n+1}^{k,l+1}=Pf(d_0,c_0,0,\cdots,2n+2,z)^{k,l}+\sigma_{2n}^{k+1,l}Q_{2n}^{k+1,l}+z\tau_{2n}^{k+1,l}P_{2n}^{k+1,l}-\tau_{2n}^{k+1,l+1}.\label{psop_conv_acc_tau_evo_odd}
%\end{align}
 \begin{align*}
&z\tau_{2n}^{k+1,l}P_{2n}^{k+1,l}=Pf(1,\cdots,2n+1,z)^{k,l},\\
&(z-1)\tau_{2n}^{k,l+1}P_{2n}^{k,l+1}=Pf(c_0,0,\cdots,2n+1,z)^{k,l},\\
&z\tau_{2n-1}^{k+1,l}P_{2n-1}^{k+1,l}-Pf(d_0,1,\cdots,2n,z)^{k,l}+\tau_{2n}^{k,l}P_{2n}^{k,l}-\tau_{2n}^{k+1,l},\\
&(z-1)\tau_{2n}^{k,l+1}P_{2n}^{k,l+1}=-Pf(c_0,d_0,1,\cdots,2n+1,z)^{k,l}+\sigma_{2n}^{k+1,l}Q_{2n}^{k+1,l}+z\tau_{2n}^{k+1,l}P_{2n}^{k+1,l}-\tau_{2n}^{k+1,l+1}.
\end{align*}
 \end{lemma}

Then, we get some relations between $P_n^{k,l}$ and $Q_{2n}^{k,l}$.
\begin{coro}
There holds the following relation between the polynomials $\{P_{n}^{k,l}\}_{n=0}^\infty$ and $\{Q_{2n}^{k,l}\}_{n=0}^\infty$:
\begin{subequations} \label{psop_conv_acc_rel_adj}
\begin{align}
&\sigma_{2n}^{k,l}\tau_{2n}^{k+1,l}Q_{2n}^{k+1,l}=(\sigma_{2n}^{k+1,l}-\tau_{2n}^{k,l+1}+\tau_{2n}^{k+1,l})\tau_{2n}^{k,l}P_{2n}^{k,l}(x)
\nonumber\\
&-\tau_{2n}^{k,l+1}(z\tau_{2n-1}^{k+1,l}P_{2n-1}^{k+1,l}-\tau_{2n}^{k,l}P_{2n}^{k,l}+\tau_{2n}^{k+1,l})+z(z-1)\tau_{2n+1}^{k,l}\tau_{2n-2}^{k+1,l+1}P_{2n-2}^{k+1,l+1},\\
&\tau_{2n+2}^{k,l}(\sigma_{2n}^{k+1,l}Q_{2n}^{k+1,l}+z\tau_{2n}^{k+1,l}P_{2n}^{k+1,l}-(z-1)\tau_{2n}^{k,l+1}P_{2n}^{k,l+1}-\tau_{2n}^{k+1,l+1})=\nonumber\\
&z\sigma_{2n+2}^{k,l}\tau_{2n}^{k+1,l}P_{2n}^{k+1,l}
-\tau_{2n}^{k+1,l+1}\tau_{2n+1}^{k,l}P_{2n+1}^{k,l}+(z-1)(\tau_{2n+1}^{k+1,l}-\tau_{2n+2}^{k,l})\tau_{2n}^{k,l+1}P_{2n}^{k,l+1}.
\end{align}
\end{subequations}
\end{coro}
\begin{proof} By employ  Lemma \ref{lem:conv_acc1}, we can rewrite \eqref{psop_conv_acc_rel_adj} in terms of uniform index on $k,l$ as
\begin{align*}
&Pf(c_0,d_0,0,1,\cdots,2n,x)Pf(1,\cdots,2n)=Pf(c_0,d_0,1,\cdots,2n)Pf(0,1,\cdots,2n,x)\\
&\ Pf(c_0,1,\cdots,2n,0)Pf(d_0,1,\cdots,2n,x)+Pf(c_0,1,\cdots,2n,x)Pf(d_0,0,1,\cdots,2n),\\
&Pf(c_0,d_0,1,\cdots,2n+1,x)Pf(0,1,\cdots,2n+1)=Pf(c_0,d_0,0,1,\cdots,2n+1)Pf(1,\cdots,2n+1,x)\\
&\ -Pf(d_0,0,1,\cdots,2n+1,x)Pf(c_0,1,\cdots,2n+1)+Pf(c_0,0,1,\cdots,2n+1,x)Pf(d_0,1,\cdots,2n+1),
\end{align*}
which are nothing but Pfaffian identities in \eqref{pf1} and \eqref{pf2}. Thus we complete the proof.
\end{proof}

Comparing the coefficients of the highest degree on the two sides of \eqref{psop_conv_acc_rel_adj} and noticing that
$$\tau_{2n+1}^{k,l}=\tau_{2n+2}^{k,l-1},$$
we obtain
\begin{subequations}\label{pfa_algo_bilinear}
\begin{align}
&\sigma_{2n}^{k,l}\tau_{2n}^{k+1,l}-\sigma_{2n}^{k+1,l}\tau_{2n}^{k,l}=\tau_{2n}^{k+1,l}\tau_{2n}^{k,l}+\tau_{2n+2}^{k,l-1}\tau_{2n-2}^{k+1,l+1}-\tau_{2n}^{k,l+1}\tau_{2n}^{k+1,l-1},\\
&\sigma_{2n+2}^{k,l}\tau_{2n}^{k+1,l}-\sigma_{2n}^{k+1,l}\tau_{2n+2}^{k,l}=\tau_{2n}^{k+1,l}\tau_{2n+2}^{k,l}+\tau_{2n}^{k+1,l+1}\tau_{2n+2}^{k,l-1}-\tau_{2n}^{k,l+1}\tau_{2n+2}^{k+1,l-1}.
\end{align}
\end{subequations}
If we introduce the variables
\begin{align*}\label{trans}
u_n^{k,l}=\frac{\tau_{2n+2}^{k,l-1}}{\tau_{2n}^{k+1,l}},\quad v_{n}^{k,l}=\frac{\sigma_{2n}^{k,l}}{\tau_{2n}^{k,l}},\quad r_n^{k,l}=\frac{\tau_{2n}^{k,l}}{\tau_{2n}^{k+1,l-1}},
\end{align*}
then $u_n^k$ $v_n^k$,  and $r_n^k$ satisfy
\begin{subequations}%\label{pfa_algo}
\begin{align*}
&r_{n+1}^{k,l}=\frac{u_n^{k,l+1}}{u_n^{k+1,l}}r_n^{k+1,l+1},\\
&v_{n+1}^{k,l}-v_n^{k+1,l}=1+\frac{u_{n}^{k,l}}{u_n^{k,l+1}}-\frac{r_{n}^{k,l+1}}{r_{n+1}^{k,l}}, \\
&\frac{u_{n+1}^{k,l}}{u_{n}^{k,l+1}}-\frac{r_{n+1}^{k,l+1}}{r_{n+1}^{k,l}}=v_{n+1}^{k,l}-v_{n+1}^{k+1,l}-1, \quad k,n,l=0,1,\ldots. 
\end{align*}
\end{subequations}
As is shown in \cite{chang2017new}, when we impose initial values on the above recurrence system by:
\begin{align*}%\label{iva}
r_{0}^{k,l}=1, \quad v_{0}^{k,l}=0, \quad u_{0}^{k,l}=\Delta^{2l}S_{k}, \quad k,l=0,1,\ldots,
\end{align*}
it can be used to accelerate the convergence of some kind of given sequence $\{S_k\}_{k=0}^\infty$. Please refer \cite{chang2017new} for more details on the convergence acceleration algorithm for sequence transformations involving Pfaffians.

\section{Conclusion and discussions}\label{sec:conc}
In the present paper, we have introduced the concept of the PSOPs together with some specific skew-symmetric integral kernels and weight functions. Appropriately deforming the weight leads to nine different types of integrable lattices together with their Pfaffian tau-functions, which also have multiple integral representations. Two of the derived integrable systems can be regarded as algorithms for certain convergence acceleration of sequence transformations and vector Pad\'e approximations. This gives the first connection between the theory of integrable systems and vector Pad\'e approximations.

It is an unexpected thing for us to obtain so many interesting integrable lattices by deforming the PSOPs. There are several directions that can be pursued for future work.  For example, one direction is to explore more derivatives of PSOPs and to produce more integrable hierarchies with the help of the PSOPs and their variants. Second, the associated numerical algorithms remain to be investigated more deeply. Furthermore,  the PSOPs and the related integrable systems may be helpful for studying more matrix models. 

%At last, we claim that discrete counterparts of the B-Toda lattices can be constructed by using the technique associated to PSOPs in this paper. However the skew-symmetric kernel is  

 \section{Acknowledgements}
This work was supported in part by the National Natural Science Foundation of China \# 11331008, 11371251, 11571358, 11701550, 11731014. We thank Professors Jonathan Nimmo and  Junxiao Zhao for their helpful discussions on full-discrete B-Toda lattices.
 \begin{appendix}

\section{On the Pfaffian}\label{app_pf} 
We include some useful materials on Pfaffians in this appendix. Let's begin with the following definition.
\begin{define}
Given a skew-symmetric matrix of order $2N$: $A=(a_{i,j})_{i,j=1}^{2N}$,  a Pfaffian of order $N$ is defined by means of the formula
\begin{align}
Pf(A)=\sum_{P}(-1)^Pa_{i_1,i_2}a_{i_3,i_4}\cdots a_{i_{2N-1},i_{2N}}.
\end{align}
The summation means the sum over all possible combinations of pairs selected from $\{1,2,\cdots,2N\}$ satisfying
\begin{align*}
&i_{2l-1}<i_{2l+1},\qquad i_{2l-1}<i_{2l}.
\end{align*}
The factor $(-1)^P$ takes the value $+1 (-1)$ if the sequence $i_1, i_2, . . . , i_{2N}$ is an even (odd) permutation of $1, 2, . . . , 2N.$ Usually, there are the conventions that the Pfaffian of order $0$ is 1 and that for negative order is $0$.
\end{define}
\begin{remark}
It has been shown rigorously by Muir \cite{muir1882treatise} that there holds 
\begin{align} \label{det_pf1}
(Pf(A))^2=\det(A)
\end{align} 
under the above definition.
\end{remark} 
\begin{remark}
There exist many notations for the Pfaffian $Pf(A)$. For our convenience, we would like to use the notation due to Hirota \cite{hirota2004direct}, that is, the notation
$$Pf(1,2,\cdots,2N)$$
based on the Pfaffian entries $Pf(i,j)$ denotes the Pfaffian $Pf(A)$ for the skew-symmetric matrix $A=(a_{i,j})_{i,j=1}^{2N}$ with $a_{i,j}=Pf(i,j)$. 
\end{remark}  
\begin{remark}\label{rem:pf2}
On many occasions,  there appear not only numbers such as $1,2,..$, but also some characters $a,b,...$ (or $1^*,2^*...$) in the expressions of the Pfaffian. It will be clear if one get the corresponding antisymmetric matrix in one's mind. For instance,  
\[Pf(a_1,a_2,1,2,\cdots,2N-2)=Pf(B),\] 
where $B=(b_{i,j})$ is the antisymmetric matrix of order $2N\times 2N$ with the entries defined by
\[ b_{ij}=\left\{\begin{array}{ll}
Pf(a_1,a_2),& i=1,\ j=2,\\
Pf(a_1,j-2), & i=1,\ 3\leq j\leq 2N,\\
Pf(a_2,j-2),& i=2,\ 3\leq j\leq 2N,\\
Pf(i-2,j-2),& 3\leq i, j\leq 2N.
\end{array}\right.
\]
Another example is 
\[Pf(1,2,\cdots,N,N^*,\cdots,2^*,1^*)=Pf(C),\] 
where
$C=(c_{i,j})$ is the antisymmetric matrix of order $2N\times 2N$ with the entries defined by
\[ c_{ij}=\left\{\begin{array}{ll}
Pf(i,j),& 1\leq i,j\leq N,\\
Pf(i,(2N+1-j)^*), & 1\leq i\leq N,\ N+1\leq j\leq 2N,\\
Pf((2N+1-i)^*,(2N+1-j)^*),& \ N+1\leq i,j\leq 2N.
\end{array}\right.
\]
\end{remark}

A Pfaffian has some basic properties similar to those for a determinant. Before we present these properties, it is necessary to point out that a row ( or column) of a Pfaffian will be imaged as that for the corresponding antisymmetric matrix. 
\begin{prop}
\begin{enumerate}
\item Multiplication of a row and a column by a constant is equivalent to multiplication of the Pfaffian by the same constant.
That is, 
$$Pf(1,2\cdots,c\cdot i,\cdots,2N)=c\cdot Pf(1,2\cdots,i,\cdots,2N).$$
\item Simultaneous interchange of the two different rows and corresponding columns changes the sign of the Pfaffian.
That is, 
$$Pf(1,2\cdots,i,\cdots,j,\cdots,2N)=-Pf(1,2\cdots, j,\cdots,i,\cdots,2N).$$
\item A multiple of a row and corresponding column added to another row and corresponding column does not change the value of the Pfaffian.
That is, 
$$Pf(1,2\cdots,i+c\cdot j,\cdots,j,\cdots,2N)=Pf(1,2\cdots, i,\cdots,j,\cdots,2N).$$
\item 
A Pfaffian admits similar Laplace expansion, that is, for a fixed $i$, $1\leq i\leq 2N$,
\begin{eqnarray*}
Pf(1,2,\ldots 2N)&=&\sum_{1\leq j\leq 2N,j\neq i}(-1)^{i+j-1}Pf(i,j)Pf(1,\cdots,\hat{i},\cdots,\hat{j},\cdots,2N),
\end{eqnarray*}
where $\hat{j}$ denotes that the index $j$ is omitted. 
If we choose $i$ as $1$ or $2N$, we have
\begin{eqnarray*}
Pf(1,2,\ldots 2N)
&=&\sum_{j=2}^{2N}(-1)^jPf(1,j)Pf(2,3,\ldots,\hat{j},\ldots,2N)\\
&=&\sum_{j=1}^{2N-1}(-1)^{j+1}Pf(1,2,\ldots,\hat{j},\ldots,2N-1)Pf(j,2N).
\end{eqnarray*}
Furthermore, if $Pf(a_0,b_0)=0,$ we have the expansion
\begin{align*}
Pf(a_0,b_0,1,2,\cdots,2N)=\sum\limits_{1\leq j<k\leq 2N}Pf(a_0,b_0,j,k)Pf(1,2,\cdots,\hat j,\cdots,\hat k,\cdots,2N).
\end{align*}

\end{enumerate} 
\end{prop}

\begin{remark}
\textit{Pfaffian} is a more general algebraic tool than \textit{determinant}. In fact, every N-order determinant can be expressed by a Pfaffian. Given a determinant
\begin{align*}
B=|(b_{i,j})|_{N\times N},
\end{align*}
then it can be expressed by a Pfaffian
$$
B=Pf (1,2,\cdots,N,N^*,\ldots,2^*,1^*),
$$
where the Pfaffian entries $Pf(i,j)$, $Pf(i^*,j^*)$, $Pf(i,j^*)$ are defined by
\begin{align*}
Pf(i,j)=Pf(i^*,j^*)=0,\ \ Pf(i,j^*)=b_{i,j},\ 1\leq i,j\leq N.
\end{align*}
Conversely, it fails.
\end{remark}

\subsection{Formulae related to determinants and Pfaffians}
There are many interesting connections between some determinants and Pfaffians. First,  we need the following basic facts:
\begin{enumerate}[I.]
\item
 For a skew-symmetric matrix $A_{2n-1}$ of size $2n -1$ augmented with an arbitrary row and column, there holds 
 \begin{align}\label{det_pf_odd}
\det\left(
\begin{array}{ccc}
A_{2n-1}&\vline&
\begin{array}{c}
x_1\\
\vdots\\
x_{2n-1}
\end{array}\\
\hline
\begin{array}{ccc}
-y_1&\cdots&-y_{2n-1}
\end{array}
&\vline&z
\end{array}
\right)=Pf(B_{2n})Pf(C_{2n}),
 \end{align}
 where 
 $$
 B_{2n}=\left(
\begin{array}{ccc}
A_{2n-1}&\vline&
\begin{array}{c}
x_1\\
\vdots\\
x_{2n-1}
\end{array}\\
\hline
\begin{array}{ccc}
-x_1&\cdots&-x_{2n-1}
\end{array}
&\vline&0
\end{array}
\right),\quad  C_{2n}=\left(
\begin{array}{ccc}
A_{2n-1}&\vline&
\begin{array}{c}
y_1\\
\vdots\\
y_{2n-1}
\end{array}\\
\hline
\begin{array}{ccc}
-y_1&\cdots&-y_{2n-1}
\end{array}
&\vline&0
\end{array}
\right).
$$

\item For a skew-symmetric matrix $A_{2n}$ of size $2n$ augmented with an arbitrary row and column, there holds 
 \begin{align}\label{det_pf_even}
\det\left(
\begin{array}{ccc}
A_{2n}&\vline&
\begin{array}{c}
x_1\\
\vdots\\
x_{2n}
\end{array}\\
\hline
\begin{array}{ccc}
-y_1&\cdots&-y_{2n}
\end{array}
&\vline&z
\end{array}
\right)=Pf(A_{2n})Pf(B_{2n+2}),
 \end{align}
 where 
 $$
 B_{2n+2}=\left(
\begin{array}{cccc}
A_{2n}&\vline&
\begin{array}{cc}
x_1&y_1\\
\vdots&\vdots\\
x_{2n}&y_{2n}
\end{array}\\
\hline
\begin{array}{ccc}
-x_1&\cdots&-x_{2n}\\
-y_1&\cdots&-y_{2n}
\end{array}
&\vline&
\begin{array}{cc}
0&z\\
-z&0
\end{array}
\end{array}
\right).
$$
\end{enumerate}

Furthermore, The results below due to de Bruijn \cite{de1955some} are useful for us. Note that the integral interval $[a,b]$ can be arbitrarily ordered set.

\begin{enumerate}[I.]
\item Let $s(x,y)$ satisfy $s(x,y)=-s(y,x)$ and the matrix $S(x_1,\cdots,x_{2N})$ is skew-symmetric with the entries $s(x_i,x_j)$. Then there holds 
%the following equalities.For the even case,
\begin{align}
\idotsint\limits_{a<x_1<\cdots<x_{2N}<b} Pf(S(x_1,\cdots,x_{2N}))\det[\varphi_i(x_j)]_{i,j=1,\cdots,2N}dx_1\cdots dx_{2N}=Pf(1,2,\ldots, 2N), \label{id_deBr1}
\end{align}
where
$$
Pf(i,j)={\iint\limits_{a<x,y<b}}\varphi_i(x)\varphi_j(y)s(x,y) dx dy.
$$
If we introduce a skew-symmetric matrix 
$$ S(x_1,\cdots,x_{2N+1})=
\left(
\begin{array}{ccc}
0&\vline&
\begin{array}{ccc}
1&\cdots&1
\end{array}
\\
\hline
\begin{array}{c}
-1\\
\vdots\\
-1
\end{array}
&\vline&s(x_i,x_j)
\end{array}
\right)_{(2N+2)\times (2N+2)},$$ then we have
\begin{align}
\idotsint\limits_{a<x_1<\cdots<x_{2N+1}<b}  Pf(S(x_1,\cdots,x_{2N+1}))\det[\varphi_i(x_j)]_{i,j=1,\cdots,2N+1}dx_1\cdots dx_{2N+1}=Pf(d_0,1,2,\ldots, 2N+1), \label{id_deBr2}
\end{align}
where
$$
Pf(d_0,i)=\int_a^b\varphi_i(y)dy, \qquad Pf(i,j)={\iint\limits_{a<x,y<b}}\varphi_i(x)\varphi_j(y)s(x,y) dx dy.
$$
\item Another integral formula reads:
\begin{align}
\int_a^b\cdots\int_a^b  \det[\varphi_i(x_j),\psi_i(x_j)]_{i=1,\cdots,2N;j=1,\cdots,N}dx_1\cdots dx_{N}=N!\ Pf(1,2,\ldots, 2N), \label{id_deBr3}
\end{align}
where
$$
Pf(i,j)={\int_a^b}\varphi_i(x)\psi_j(x)-\varphi_j(x)\psi_i(x)dx.
$$

\end{enumerate}
%\item There is an generalization in  \cite{de1955some} for the de Bruijn's formulae  \eqref{id_deBr1}-\eqref{id_deBr2}, which we shall call generalized de Bruijn's formulae.

\subsection{Formulae for special Pfaffians}
 The following formulae play important roles in our paper. 
\begin{enumerate}[I.]

\item
 A central role is played by Schur's Pfaffian identity\cite{ishikawa2006generalizations,schur1911uber} in enumerative combinatorics and representation theory, which reads
\begin{align}
Pf(1,2,\ldots, 2N)=\prod_{1\leq i<j\leq 2N}\frac{s_i-s_j}{s_i+s_j}, \qquad \text{where} \qquad Pf(i,j)=\frac{s_i-s_j}{s_i+s_j},\label{id_schur}
\end{align}
for even case. In the odd case, if $Pf(a_0,i)=1$, we have 
\begin{align}
Pf(a_0,1,2,\ldots, 2N-1)=\prod_{1\leq i<j\leq 2N-1}\frac{s_i-s_j}{s_i+s_j}, \qquad \text{where} \qquad Pf(i,j)=\frac{s_i-s_j}{s_i+s_j}.\label{id_schur_odd}
\end{align}

\item
There hold some derivative formulae for some special Pfaffians \cite{hirota2004direct}.
\begin{enumerate}[(1)]
\item If the $x$-derivative of a Pfaffian entry $Pf(i,j)$ satisfy $\frac{\partial}{\partial x}Pf(i,j)=Pf(i+1,j)+Pf(i,j+1)$, then 
\begin{align}
\frac{\partial}{\partial t}Pf(i_1,\cdots,i_{2N})&=\sum_{k=1}^{2N}Pf(i_1,i_2,\cdots,i_k+1,\cdots,i_{2N}),\label{der1_gen}
\end{align}
which gives as a special case
\begin{align}
\frac{\partial}{\partial x}Pf(1,2,\cdots,2N)&=Pf(1,2,\cdots,2N-1,2N+1) .\label{der1}
\end{align}
If $\frac{\partial}{\partial x}Pf(a_0,i)=Pf(a_0,i+1),$ then
\begin{align}
\frac{\partial}{\partial x}Pf(a_0,1,2,\cdots,2N-1)&=Pf(a_0,1,2,\cdots,2N-2,2N) .\label{der1_odd}
\end{align}

%Similarly, if $\frac{\partial}{\partial x}(i,j)=(i-1,j)+(i,j-1)$, then 
%\begin{align}
%\frac{\partial}{\partial t}(i_1,\cdots,i_{2N})&=\sum_{k=1}^{2N}(i_1,i_2,\cdots,i_k-1,\cdots,i_{2N}),\label{der1_2_gen}
%\end{align}
%which gives as a special case
%\begin{align}
%\frac{\partial}{\partial x}(1,2,\cdots,2N)&=(0,2,3,\cdots,2N) .\label{der1_2}
%\end{align}
%If $\frac{\partial}{\partial x}(a_0,i)=(a_0,i-1),$ then
%\begin{align}
%\frac{\partial}{\partial x}(a_0,1,2,\cdots,2N-1)&=(a_0,0,2,\cdots,2N-1) .\label{der2_odd}
%\end{align}

\item If  $\frac{\partial}{\partial x} Pf(i,j)=Pf(a_0,b_0,i,j)$ and $Pf(a_0,b_0)=0$, then
    \begin{align}
        \frac{\partial}{\partial x} Pf(i_1,i_2,\ldots,i_{2N})=Pf(a_0,b_0,i_1,i_2,\ldots,i_{2N}).\label{der2_1}
    \end{align}
%which gives as a special case
%\begin{align}
%    \frac{\partial}{\partial x} Pf(1,2,\ldots,2N)=Pf(a_0,b_0,1,2,\ldots,2N).\label{der2_1}
%    \end{align}
    If $\frac{\partial}{\partial x} Pf(i,j)=Pf(a_0,b_0,i,j)$ , $\frac{\partial}{\partial x} Pf(a_0,j)=Pf(b_0,j)$, then
     \begin{align}
        \frac{\partial}{\partial x} Pf(a_0,i_1,i_2,\ldots,i_{2N-1})=Pf(b_0,i_1,i_2,\ldots,i_{2N-1}).\label{der2_2}
    \end{align}
%which gives as a special case
%      \begin{align}
%    \frac{\partial}{\partial x} Pf(a_0,1,2,\ldots,2N-1)=Pf(b_0,1,2,\ldots,2N-1).\label{der2_2}
%    \end{align}
\end{enumerate}

\begin{remark}
The first type formula can be named as the Wronski-type formula  since it has exactly the same form as the differential rule for the Wronskian determinant. The second type is referred as Gram-type formula. See \cite{hirota2004direct} for more details. The discrete counterparts of Gram-type are listed in \eqref{add_g1}-\eqref{add_g2}, and  Wronski-type in \eqref{add_w1}-\eqref{add_w2}.
\end{remark}
\item There hold the discrete Gram-type formulae (also called addition formulae) \cite{hirota2004direct,hirota2013additional,ohta2004special}.
%\begin{enumerate}[(1)]
%\item 

If $Pf(i^*,j^*)=Pf(i,j)+\lambda Pf(a_0,b_0,i,j)$ and $Pf(a_0,b_0)=0$, then \footnote{In the main body, $\lambda$ is taken as $1$ or $-1$.}
\begin{align}
Pf(1^*,2^*,\cdots,(2N)^*)= Pf(1,2,\ldots,2N)+\lambda Pf(a_0,b_0,1,2,\ldots,2N).\label{add_g1}
\end{align}

 If $Pf(a_0,i^*)=Pf(a_0,i)+\lambda Pf(b_0,i)$, then 
\begin{align}
Pf(a_0,1^*,2^*,\cdots,(2N-1)^*)= Pf(a_0,1,2,\ldots,2N-1)+\lambda Pf(b_0,1,2,\ldots,2N-1).\label{add_g2}
\end{align}
%\end{enumerate}
\item The following discrete Wronski-type formulae can be found in \cite{miki2011discrete,ohta2004special}.
%\begin{enumerate}[(1)]
%\item 

If $Pf(i^*,j^*)=\lambda^2Pf(i,j)+\lambda Pf(i+1,j)+ \lambda Pf(i,j+1)+Pf(i+1,j+1)$, then 
\begin{align}
Pf(1^*,2^*,\cdots,(2N)^*)= Pf(c_0,1,2,\ldots,2N+1).\label{add_w1}
\end{align}
In addition, if $Pf(a_0,i^*)=\lambda Pf(a_0,i)+(a_0,i+1)$, then 
\begin{align}
Pf(a_0,1^*,2^*,\cdots,(2N-1)^*)= Pf(a_0,c_0,1,2,\ldots,2N).\label{add_w2}
\end{align}
Here $(c_0,i)=(-\lambda )^{i-1}, (a_0,c_0)=0.$
%\end{enumerate}
\end{enumerate}
\begin{remark}
Usually, all of the above formulae for the general cases may be derived by induction. In the main body, we will use the induction to develop some extra formulae.
\end{remark}

\subsection{Bilinear identities on Pfaffians}
As is known, there exist many identities in the theory of determinants. As a more generalized mathematical tool than determinants, Pfaffians also allow various kinds of identities. We will give two common Pfaffian identities in the following, which can be found in \cite{hirota2004direct}:
%\begin{align*}
%&(a_1,a_2,\ldots,a_{2k},1,\ldots,2N)(1,2,\ldots,2N)\nonumber\\
%=&\sum_{j=2}^{2k}(-1)^j(a_1,a_j,1,\ldots,2N)(a_2,a_3,\ldots,\hat{a_j},\ldots,a_{2k},1,\ldots,2N),\\
%&(a_1,a_2,\ldots,a_{2k-1},1,\ldots,2N-1)(1,2,\ldots,2N)\nonumber\\
%=&\sum_{j=1}^{2k-1}(-1)^{j-1}(a_j,1,\ldots,2N-1)(a_1,a_2,\ldots,\hat{a_j},\ldots,a_{2k-1},1,\ldots,2N).
%\end{align*}
%In the particular case of $k=2$, the above identities respectively become
\begin{align}
&Pf(a_1,a_2,a_3,a_4,1,\ldots,2N)Pf(1,2,\ldots,2N)\nonumber\\
=&\sum_{j=2}^{4}(-1)^jPf(a_1,a_j,1,\ldots,2N)Pf(a_2,\hat{a}_j,a_{4},1,\ldots,2N),\label{pf1}\\
&Pf(a_1,a_2,a_3,1,\ldots,2N-1)Pf(1,2,\ldots,2N)\nonumber\\
=&\sum_{j=1}^{3}(-1)^{j-1}Pf(a_j,1,\ldots,2N-1)Pf(a_1,\hat{a}_j,a_{3},1,\ldots,2N).\label{pf2}
\end{align} 

Let us mention that Jacobi identity and Pl\"{u}cker relation \cite{aitken1959determinants,hirota2004direct} may be viewed as the particular cases of the above Pfaffian identities.

 \section{SOPs and Pfaff lattices}\label{sec:sop} 
 In this appendix, we would like to give a brief introduction on the SOPs and the associated integrable systems for the convenience of the readers.  Let's begin with the following definition.
 \begin{define}
Let $\langle \cdot, \cdot\rangle$ be a skew-symmetric inner product in the polynomial space over the field of real numbers, more exactly speaking, a bilinear 2-form from $\mathbb{R}(z)\times\mathbb{R}(z)\rightarrow \mathbb{R}$ satisfying the skew symmetric relation:
$$\langle f(z),g(z)\rangle=-\langle g(z),f(z)\rangle.$$
A family of (monic) polynomials $\{P_n(z)\}_{n=0}^{\infty}$ are called SOPs  with respect to the  skew-symmetric inner product $\langle \cdot, \cdot\rangle$ if they satisfy the following relations:
\begin{subequations}\label{skew_inner}
\begin{align}
&\langle P_{2n}(z), P_{2m+1}(z)\rangle=r_n\delta_{nm},\label{skew_inner1}\\
&\langle P_{2n}(z), P_{2m}(z)\rangle=\langle P_{2n+1}(z), P_{2m+1}(z)\rangle=0,\label{skew_inner2}
\end{align}
\end{subequations}
for some appropriate nonzero numbers $r_n$. 
 \end{define}
 \begin{remark}
 From the orthogonality \eqref{skew_inner}, the SOPs of odd degree have the following ambiguity: the skew orthogonality relation \eqref{skew_inner} is invariant under the replacement
$$
P_{2n+1}(z)\rightarrow P_{2n+1}(z)+\alpha_nP_{2n}(z)
$$
for any number $\alpha_n$.
In many cases, the coefficient of $z^{2n}$ in $P_{2n+1}(z)$ is usually chosen as 0 and then the SOPs are uniquely determined under some appropriate condition (as will be seen in the following subsections).
 \end{remark}
\begin{remark}
The partition functions of random matrices for Gaussian orthogonal ensembles (GOE) and Gaussian symplectic ensembles (GSE) are associated with SOPs with respect to  skew-symmetric inner products 
\begin{eqnarray*}
&&\langle f(z),g(z)\rangle_{goe}=\int_{-\infty}^{\infty}\int_{-\infty}^{\infty}sgn(x-y)f(x)g(y)e^{-NV(x)-NV(y)}dxdy,\\
&&\langle f(z),g(z)\rangle_{gse}=\int_{-\infty}^{\infty}\left(f(x)g'(x)-f'(x)g(x)\right)e^{-NV(x)}dx,
\end{eqnarray*}
respectively. Here $N$ is the order of the random matrix and $V(x)$ is an even degree polynomial with positive leading coefficient. 
\end{remark}

\subsection{Representations in terms of determinants}
The SOPs can be explicitly expressed in terms of the bi-moments defined by
\begin{align}
\mu_{i,j}=\langle z^i,z^j\rangle=-\langle z^j,z^i\rangle
\end{align}
under some nondegenerate assumption. If the determinant $\det\left((\mu_{i,j})_{i,j=0}^{2n-1}\right)$ of coefficient matrix is nonzero, according to the linear system obtained from \eqref{skew_inner}, it is not hard to get
\begin{subequations}
\begin{align}
 &P_{2n}(z)=\frac{1}{\det\left(\mu_{i,j}\right)_{0\leq i,j\leq 2n-1}}
 \left|
 \begin{array}{cccc}
 \mu_{0,0}& \mu_{0,1}&\cdots& \mu_{0,2n}\\
  \mu_{1,0}& \mu_{1,1}&\cdots& \mu_{1,2n}\\
  \vdots&\vdots&\ddots&\vdots\vdots\\
   \mu_{2n-1,0}& \mu_{2n-1,1}&\cdots& \mu_{2n-1,2n}\\
   1&z&\cdots&z^{2n}
 \end{array}
 \right|,
\label{sop_even_det}\\
& P_{2n+1}(z)=\frac{1}{\det\left(\mu_{i,j}\right)_{0\leq i,j\leq 2n-1}}
 \left|
 \begin{array}{ccccc}
 \mu_{0,0}& \mu_{0,1}&\cdots& \mu_{0,2n-1}&\mu_{0,2n+1}\\
  \mu_{1,0}& \mu_{1,1}&\cdots& \mu_{1,2n-1}&\mu_{1,2n+1}\\
  \vdots&\vdots&\ddots&\vdots&\vdots\\
     \mu_{2n-1,0}& \mu_{2n-1,1}&\cdots& \mu_{2n-1,2n-1}&\mu_{2n-1,2n+1}\\
      1&z&\cdots&z^{2n-1}&z^{2n+1}
 \end{array}
 \right|.\label{sop_odd_det}
\end{align}
\end{subequations}
From the fact that the determinant of any skew-symmetric matrix of odd order is zero, it obviously follows that
\begin{align}
&\langle P_{2n}(z), P_{2n}(z)\rangle=\langle P_{2n}(z), z^{2n}\rangle=\langle z^{2n}, P_{2n}(z)\rangle=0,\\
&\langle P_{2n+1}(z),P_{2n+1}(z)\rangle=\langle P_{2n+1}(z),z^{2n+1}\rangle=\langle z^{2n+1}, P_{2n+1}(z)\rangle=0.
\end{align}
 And, 
the $r_n$ in \eqref{skew_inner1} is determined by 
\begin{align}
r_n&=\langle P_{2n}(z), P_{2n+1}(z)\rangle=\langle P_{2n}(z), z^{2n+1}\rangle\nonumber\\
&=
\frac{1}{\det\left(\mu_{i,j}\right)_{0\leq i,j\leq 2n-1}}
 \left|
 \begin{array}{ccccc}
 \mu_{0,0}& \mu_{0,1}&\cdots& \mu_{0,2n-1}&\mu_{0,2n+1}\\
  \mu_{1,0}& \mu_{1,1}&\cdots& \mu_{1,2n-1}&\mu_{1,2n+1}\\
  \vdots&\vdots&\ddots&\vdots&\vdots\\
   \mu_{2n,0}& \mu_{2n,1}&\cdots& \mu_{2n,2n-1}&\mu_{2n,2n+1}\\
 \end{array}
 \right|.\label{exp:r}
\end{align}
Therefore, we have obtained the explicit expressions in terms of the determinants associated to the moments as \eqref{sop_even_det} and \eqref{sop_odd_det}. Actually, the skew-symmetric property of the moment matrix $(\mu_{i,j})_{i,j=0}^{k}$ promotes one to give the representations in terms of Pfaffians, which was noticed by Adler, Horozov and van Moerbeke in \cite{adler1999pfaff}.

\subsection{Representations in terms of Pfaffians}  
Now we are ready to give the representations of SOPs in terms of Pfaffians. The readers may refer \cite{adler1999pfaff,miki2011discrete} for more information.
\begin{theorem} If any Pfaffian 
$$Pf(0,1,\cdots,2n-1)\triangleq\tau_{2n}$$ is nonzero, then the monic SOPs $\{P_n(z)\}_{n=0}^{\infty}$ have the following explicit form in terms of Pfaffians 
\begin{subequations}
\begin{align}
&P_{2n}(z)=\frac{1}{\tau_{2n}}Pf(0,1,\cdots,2n-1,2n,z),\\
&P_{2n+1}(z)=\frac{1}{\tau_{2n}}Pf(0,1,\cdots,2n-1,2n+1,z).
\end{align}
\end{subequations}
Here the Pfaffian entries are defined as 
$$Pf(i,j)=\mu_{i,j},\qquad Pf(i,z)=z^i.$$
In this case, 
$$r_n=\langle P_{2n}(z), P_{2n+1}(z)\rangle=\langle P_{2n}(z), z^{2n+1}\rangle=\frac{\tau_{2n+2}}{\tau_{2n}}.$$
\end{theorem}
\begin{proof}
Looking at the denominators in the representations \eqref{sop_even_det} and \eqref{sop_odd_det}, it is obvious that 
$$\det\left(\mu_{i,j}\right)_{0\leq i,j\leq 2n-1}=\tau_{2n}^2.$$
 
It is sufficient to prove the following equalities 
\begin{align*}
&\left|
 \begin{array}{ccccc}
 0& \mu_{0,1}&\cdots& \mu_{0,2n-1}&\mu_{0,2n}\\
  -\mu_{0,1}& 0&\cdots& \mu_{1,2n-1}&\mu_{1,2n}\\
  \vdots&\vdots&\ddots&\vdots&\vdots\\
   -\mu_{0,2n-1}& -\mu_{1,2n-1}&\cdots&0&\mu_{2n-1,2n}\\
   1&z&\cdots&z^{2n-1}&z^{2n}
 \end{array}
 \right|=\tau_{2n}Pf(0,1,\cdots,2n-1,2n,z),\\
 &\left|
 \begin{array}{ccccc}
 0& \mu_{0,1}&\cdots& \mu_{0,2n-1}&\mu_{0,2n+1}\\
  -\mu_{0,1}& 0&\cdots& \mu_{1,2n-1}&\mu_{1,2n+1}\\
  \vdots&\vdots&\ddots&\vdots&\vdots\\
     -\mu_{0,2n-1}& -\mu_{1,2n-1}&\cdots& 0&\mu_{2n-1,2n+1}\\
      1&z&\cdots&z^{2n-1}&z^{2n+1}
 \end{array}
 \right|=\tau_{2n}Pf(0,1,\cdots,2n-1,2n+1,z).
%  \left|
% \begin{array}{cc}
% (\mu_{i,j})_{i,j=0}^{2n-1}&
% \begin{array}{c}
% \mu_{0,2n+1}\\
% \mu_{1,2n+1}\\
% \vdots\\
% \mu_{2n-1,2n+1}
% \end{array}\\
% \begin{array}{cccc}
% 1&z&\cdots&z^{2n-1}
% \end{array}
% &z^{2n+1}
% \end{array}
% \right|
 \end{align*}
Actually these two equalities are immediately obtained by applying \eqref{det_pf_odd} and \eqref{det_pf_even}.

 Moreover,  a direct computation by use of the formula for expansions of Pfaffians will lead to
  \begin{align*}
  r_n&=\langle P_{2n}(z), z^{2n+1}\rangle=\frac{1}{\tau_{2n}}\langle Pf(0,1,\cdots,2n-1,2n,z), z^{2n+1}\rangle\\
  &=\frac{1}{\tau_{2n}} \sum_{i=0}^{2n}(-1)^iPf(0,1,\cdots,\hat i,\cdots,2n) \langle z^i, z^{2n+1}\rangle\\
  &=\frac{1}{\tau_{2n}} \sum_{i=0}^{2n}(-1)^iPf(0,1,\cdots,\hat i,\cdots,2n)Pf(i,2n+1)=\frac{\tau_{2n+2}}{\tau_{2n}}.
    \end{align*}
   \end{proof}
   \begin{remark}
   As for the GOE case, the corresponding Pfaffians $Pf(0,1,\cdots,2n-1)\triangleq\tau_{2n}$ with the Pfaffian entries $$Pf(i,j)=\mu_{i,j}=\int_{-\infty}^{\infty}\int_{-\infty}^{\infty}sgn(x-y)x^iy^je^{-NV(x)-NV(y)}dxdy$$ are strictly nonzero. In fact this can be proved by  use of the equality \eqref{id_deBr1} due to de Bruijn \cite{de1955some}, which gives 
      \begin{align*}
   \tau_{2n}&=(-1)^n\idotsint\limits_{-\infty< x_1<\cdots< x_{2n}< \infty} \det(x_j^{i-1})_{1\leq i,j\leq 2n}e^{-N\sum_{i=1}^{2n}V(x_i)}dx_1\cdots dx_{2n}\\
   &=(-1)^n\idotsint\limits_{-\infty< x_1<\cdots< x_{2n}< \infty} \prod_{1\leq i<j\leq 2n}(x_j-x_i)e^{-N\sum_{i=1}^{2n}V(x_i)}dx_1\cdots dx_{2n}\neq 0,
   \end{align*}
   which means that the corresponding SOPs make sense.
   \end{remark}
   \begin{remark}
As for the Pfaffian entries defined by    
   $$Pf(i,j)=\mu_{i,j}=\int_{-\infty}^{\infty}\left(x^i (x^j)'-x^j (x^{i})'\right)e^{-NV(x)}dx,$$
  it follows from  the de Bruijn's equality \eqref{id_deBr3} that
 \begin{align*}
 \tau_{2n}&=\frac{1}{n!}\int_{-\infty}^\infty\cdots\int_{-\infty}^\infty\det\left(x_j^{i},ix_j^{i-1}\right)_{0\leq i\leq 2n-1,1\leq j\leq n}e^{-N\sum_{i=1}^{n}V(x_i)}dx_1\cdots dx_n\\
 &=\frac{1}{n!}\int_{-\infty}^\infty\cdots\int_{-\infty}^\infty\prod_{1\leq i<j\leq n}(x_j-x_i)^4e^{-\sum_{i=1}^{n}V(x_i)}dx_1\cdots dx_n>0,
   \end{align*}
   which corresponds to the GSE case. It implies that the corresponding SOPs are well-defined.
     \end{remark}   
In the setup in terms of Pfaffians, the skew-orthonormal polynomials $\{q_k\}_{0\leq k\leq \infty}$ given by
   $$q_{2n}(z)=\frac{1}{\sqrt{\tau_{2n}\tau_{2n+2}}}Pf(0,1,\cdots,2n-1,2n,z),\quad q_{2n+1}(z)=\frac{1}{\sqrt{\tau_{2n}\tau_{2n+2}}}Pf(0,1,\cdots,2n-1,2n+1,z)$$
   satisfy
   $$\langle q_i,q_j\rangle_{i,j\geq0}=J,$$
   where  $J$ is a semi-infinite skew-symmetric matrix with zero everywhere except for the following $2\times2$ blocks along the ``diagonal'':
$$
J=\left(\begin{array}{ccccc}
0&1&&&\\
-1&0&&&\\
&&0&1&\\
&&-1&0&\\
&&&&\ddots
\end{array}\right).
$$
If we write $$q(z)=(q_0,q_1,\cdots)^\top=Q\chi(z),\qquad \chi(z)=(1,z,z^2,\cdots)^\top,$$
then the semi-infinite skew-symmetric moment matrix $\mu_{\infty}=(\mu_{i,j})_{0\leq i,j\leq\infty}$ satisfies
$$Q\mu_\infty Q^\top=J,$$
leading to
$$\mu_\infty=Q^{-1}JQ^{\top -1},$$
which is called the ``skew-Borel decomposition'' for the skew symmetric matrix $\mu_\infty$. Note that $Q$ is lower-triangular.
Furthermore,  it is not hard to see that there holds
\begin{align}
zq(z)=Lq(z),\qquad\qquad \text{with}\qquad L=Q\Lambda Q^{-1}, \label{exp:LQ}
\end{align}
where $\Lambda$ is the shift matrix.
\subsection{Pfaff lattice}
In \cite{adler1999pfaff}, Adler, Horozov and van Moerbeke constructed the Pfaff lattice,  whose tau-functions are Pfaffians and the wave vectors of the Lax pair SOPs.
We restate it as follows.

Consider the semi-infinite skew-symmetric moment matrix $\mu_{\infty}=(\mu_{i,j})_{0\leq i,j\leq\infty}$ produced by the initial skew-symmetric matrix:
$$
\mu_\infty(t)=e^{\sum_{k>0} t_k \Lambda^k} \mu_\infty(0) e^{\sum_{k>0} t_k \Lambda^{\top k}}.
$$
 It is not hard to see that $\mu_{\infty}(t)$ actually evolves according to the vector fields 
\begin{align}\label{evo:pf_mom}
\frac{\partial\mu_\infty}{\partial t_k}=\Lambda^k\mu_\infty+\mu_\infty\Lambda^{\top k},\qquad k=1,2,\cdots.
\end{align}
It is shown in \cite{adler1999pfaff} that the matrix $L$ defined in \eqref{exp:LQ} evolves according to
\begin{align} \label{eq:pff_lattice}
\frac{\partial L}{\partial t_i}=[-\pi_{\textbf{k}}L^i,L]=[\pi_{\textbf{n}}L^i,L] ,
\end{align}
which is nothing but the Pfaff lattice\footnote{Here $\textbf{k}$ is the Lie algebra of lower-triangular matrices with some special feature and 
$$\textbf{n}:= \{a \in \mathscr{D}\  \text{such that}\ Ja^\top J=a\}=sp(\infty)$$
satisfying  $\textbf{k} + \textbf{n}=\mathscr{D} := gl_\infty$.
The $\pi_{\textbf{k}}, \pi_{\textbf{n}}$ are the corresponding projection operators. Please refer \cite{adler1999pfaff} for detailed description.

%defined as follows:  Let $\mathscr{D}:=gl_\infty$. consider the splitting of $\mathscr{D} = \textbf{k} + \textbf{k}$ into two Lie subalgebras $\textbf{k}$ and $\textbf{n}$ are the corresponding projections , where $\textbf{k}$ is the Lie algebra of lower- triangular matrices with some special feature 
%Also, consider the splitting of $D = k + n$ into two Lie subalgebras k and n, with the corresponding projections denoted $\pi_k$ and $\pi_n$, where k is the Lie algebra of lower- triangular matrices with some special feature  and where
%$$n:= a \in D such that Ja\top J=a=sp(\infty)$$
}. 
It was shown to be integrable, by virtue of the AKS theorem \cite{adler1999pfaff}. 

\begin{remark}
Recall that the Toda lattice corresponds to a flow of Jacobi (tridiagonal) matrix. By contrast, any $t_k$ flow in \eqref{eq:pff_lattice} is quite an intricate system, because the matrix $L$ is a lower Hessenberg matrix and does not own concise structures. For more related work, please refer \cite{adler2002pfaff,adler2002toda,kodama2007geometry,kodama2010pfaff} etc..
\end{remark}

\subsection{Discrete counterpart of Pfaff lattice} By introducing the so-called skew-Christoffel transformation for SOPs, and its inverse transformation, Miki, Goda and Tsujimoto \cite{miki2011discrete} obtained a discrete integrable system in 1+1 dimension, which they call a discrete counterpart of the Pfaff lattice.
We restate their derivation as follows:

Introduce a series of polynomial sequences $\{P_n^t(z)\}_{n=0}^\infty$ (Here $t\in \mathbb{N}_0$ is a discrete variable)  iterated by 
\begin{align*}
&P_n^0(z):=P_n(z),\\
&P_{2n}^{t+1}(z):=\frac{1}{z-\lambda}\left(P_{2n+1}^t(z)+\sum_{k=0}^nA_{n,k}^tP_{2k}^t(z)+\sum_{k=0}^{n-1}B_{n,k}^tP_{2k+1}^t(z)\right),\\
&P_{2n+1}^{t+1}(z):=\frac{1}{z-\lambda}\left(P_{2n+2}^t(z)+C_n^tP_{2n}^t(z)\right),
\end{align*}
where $\lambda$ is some constant parameter and 
$$A_{n,k}^t=-\frac{r_n^t}{r_k^t}\frac{P_{2k+1}^t(\lambda)}{P_{2n}^t(\lambda)},\qquad B_{n,k}^t=\frac{r_n^t}{r_k^t}\frac{P_{2k}^t(\lambda)}{P_{2n}^t(\lambda)}, \qquad C_{n}^t=-\frac{P_{2n+2}^t(\lambda)}{P_{2n}^t(\lambda)}.$$

Then, for any fixed $t,$ $\{P_n^t\}_{n=0}^\infty$ are SOPs with the  skew-symmetric inner products $\langle\cdot,\cdot,\rangle^t$ produced by
\begin{align}\label{inn_CT}
\langle\cdot,\cdot\rangle^0:=\langle\cdot,\cdot\rangle,\qquad \langle\cdot,\cdot\rangle^{t+1}:=\langle(z-\lambda)\cdot,(z-\lambda)\cdot\rangle^t.
\end{align} More exactly,
$$\langle P_{2n}^t(z), P_{2m+1}^t(z)\rangle^t=r_n^t\delta_{nm},\qquad \langle P_{2n}^t(z), P_{2m}^t(z)\rangle^t=\langle P_{2n+1}^t(z), P_{2m+1}^t(z)\rangle^t=0.$$

The inverse transformation for the skew-Christoffel transformation, that is the transformation from $\{P_n^{t+1}\}_{n=0}^\infty$ to $\{P_n^t\}_{n=0}^\infty$,  can be constructed as follows 
\begin{align*}
&P_{2n}^{t}(z):=P_{2n}^{t+1}(z)+\sum_{k=0}^{n-1}\tilde A_{n,k}^tP_{2k}^{t+1}(z)+\sum_{k=0}^{n-1}\tilde B_{n,k}^tP_{2k+1}^{t+1}(z),\\
&P_{2n+1}^t(z):=P_{2n+1}^{t+1}(z)+\sum_{k=0}^{n}\tilde C_{n,k}^tP_{2k}^{t+1}(z)+\sum_{k=0}^{n-1}\tilde D_{n,k}^tP_{2k+1}^{t+1}(z),
\end{align*}
where 
\begin{align*}
&\tilde A_{n,k}^t=\frac{1}{r_k^{t+1}}\langle P_{2n}^t(z), P_{2k+1}^{t+1}(z)\rangle^{t+1},\qquad \tilde B_{n,k}^t=\frac{1}{r_k^{t+1}}\langle P_{2n}^t(z),P_{2k}^{t+1}(z)\rangle^{t+1},\\
&\tilde C_{n,k}^t=\frac{1}{r_k^{t+1}}\langle P_{2n+1}^t(z), P_{2k+1}^{t+1}(z)\rangle^{t+1},\qquad \tilde D_{n,k}^t=\frac{1}{r_k^{t+1}}\langle P_{2n+1}^t(z),P_{2k}^{t+1}(z)\rangle^{t+1}.
\end{align*}
If we introduce $\Phi^t=(P_0^t(z),P_1^t(z),\cdots)^\top$, then skew-Christoffel transformation and its inverse can be formulated as 
\begin{align*}
(z-\lambda)\Phi^{t+1}=L^t\phi^t,\qquad\qquad \Phi^{t}=R^t\Phi^{t+1},
\end{align*}
where 
$$
L^t=\left(\begin{array}{cccccc}
A_{0,0}^t&1&&&&\\
C_0^t&0&1&&&\\
A_{1,0}^t&B_{1,0}^t&A_{1,1}^t&1&&\\
0&0&C_1^t&0&1&\\
\vdots&\ddots&\ddots&\ddots&\ddots&\ddots
\end{array}
\right),\qquad
R^t=\left(\begin{array}{ccccc}
1&&&&\\
\tilde C_{0,0}^t&1&&&\\
\tilde A_{1,0}^t&\tilde B_{1,0}^t&1&&\\
\tilde C_{1,0}^t&\tilde D_{1,0}^t&\tilde C_{1,1}^t&1&\\
\vdots&\ddots&\ddots&\ddots&\ddots
\end{array}
\right).
$$
The compatibility condition of the above Lax pair yields the discrete Lax equation:
\begin{align}\label{eq:disc_pf_lat}
L^tR^t=R^{t+1}L^{t+1},
\end{align}
which is regarded as a discrete counterpart of the $t_1$ flow of the Pfaff lattice by the authors \cite{miki2011discrete} by noticing that 
\eqref{inn_CT} may lead to the moment evolution in \eqref{evo:pf_mom}
if a suitable limit is taken.
\begin{remark}
The equation \eqref{eq:disc_pf_lat} is a nonlocal discrete integrable system in $1+1$ dimension. Actually, the authors \cite{miki2011discrete} also derived a local discrete integrable system in $2+1$ dimension by introducing a 2-parameter spectral transformation.
\end{remark}

 \end{appendix}
    
\small
\bibliographystyle{abbrv}
\def\cydot{\leavevmode\raise.4ex\hbox{.}}
  \def\cydot{\leavevmode\raise.4ex\hbox{.}} \def\cprime{$'$}

\end{document}